\documentclass[12pt]{article}
\usepackage{amssymb}
\usepackage{amsbsy}
\usepackage{epsfig}
\usepackage{verbatim}
\usepackage{youngtab}
\usepackage{graphicx}
\usepackage{subfigure}
\usepackage[utf8]{inputenc}
\usepackage[colorlinks=true]{hyperref}
\usepackage[usenames,dvipsnames]{color}
\usepackage{xspace}
\usepackage{bm}
\usepackage{tikz}
\usetikzlibrary{snakes}

\makeatletter \@addtoreset{figure}{section}
\def\thefigure{\thesection.\@arabic\c@\input{../../../../Applications/TeX/TeXShop.app}
figure}
\def\fps@figure{h, t}
\@addtoreset{table}{bsection}
\def\thetable{\thesection.\@arabic\c@table}
\def\fps@table{h, t}
\@addtoreset{equation}{section}

\newtheorem{theorem}{Theorem}[section]
\newtheorem{corollary}[theorem]{Corollary}
\newtheorem{definition}[theorem]{Definition}
\newtheorem{proposition}[theorem]{Proposition}
\newtheorem{example}{Example}[section]
\newtheorem{lemma}[theorem]{Lemma}
\newtheorem{remark}{Remark}[section]
\newtheorem{remarks}[remark]{Remarks}

\def\bd{\begin{definition}}
\def\ed{\end{definition}}
\def\bt{\begin{theorem}}
\def\et{\end{theorem}}
\def\bp{\begin{proposition}\rm}
\def\ep{\end{proposition}}
\def\bc{\begin{corollary}}
\def\ec{\end{corollary}}
\def\bl{\begin{lemma}\em}
\def\el{\end{lemma}}
\def\be{\begin{equation}}
\def\ee{\end{equation}}
\def\br{\begin{remark}\rm\small}
\def\er{\end{remark}}
\def\brs{\begin{remarks}.\\ \rm\
\begin{enumerate}}
\def\ers{\end{enumerate}\end{remarks}}
\def\bea{\begin{eqnarray}}
\def\eea{\end{eqnarray}}

\newenvironment{proof}[1][Proof]{\noindent {\bf #1.} }{\hfill {\bf Q.E.D.}\\}

\newcommand{\Mat}[2]{{\mathop {\mathrm {Mat}}}^{#1 \times #2}}
\newcommand{\Id}[1]{{\mathbf{I}}_{#1}}
\newcommand{\ef}[1]{{\bm e}(#1)}
\newcommand{\eg}[1]{{\bm E}(#1)}
\def\vf{{\mathbf f}}
\def\vg{{\mathbf g}}
\def\vv{{\mathbf v}}

\def\vfe{{\mathbf f}}
\def\vge{{\mathbf g}}

\def\Cn{{\mathbf C}}
\def\Rn{{\mathbf R}}
\def\Zn{{ \mathbf Z}}
\def\Nn{{ \mathbf N}}
\def\Pr{{\mathbf P}}

\def\ra{{\rightarrow}}

\def\mt{{\mapsto}}

\def\Tr{\mathop \mathrm {Tr}}
\def\tr{\mathop \mathrm {tr}}
\def\det{\mathop \mathrm {det}}

\def\dim{\mathop \mathrm {dim}}

\def\End{\mathop \mathrm {End}}
\def\span{\mathop \mathrm {span}}

\def\diag{\mathop \mathrm {diag}}

\def\ss{\subset}

\def\Gr{\mathop \mathrm {Gr} \nolimits}

\def\ds{\displaystyle}

\def\&{&{\hskip -20pt}}

\def\AA{\mathcal{A}}

\def\FF{\mathcal{F}}

\def\HH{\mathcal{H}}

\def\MM{\mathcal{M}}

\def\Cb{\mathbf{C}}
\def\Ib{\mathbf{I}}

\def\ab{\mathbf{a}}
\def\bb{\mathbf{b}}
\def\cb{\mathbf{c}}

\def\tb{\mathbf{t}}

\def\Pb{\mathbf{P}}

\def\Zb{\mathbf{Z}}

\def\0b{\mathbf{0}}

 \def\grg{\mathfrak{g}}

 \def\grl{\mathfrak{l}}

\def\grP{\mathfrak{P}}

\def\nchi{\hbox{\raise 2.5pt\hbox{$\chi$}}}
\begin{document}
\baselineskip 16pt
\begin{flushright}
CRM-3336 (2014)
\bigskip
\end{flushright}
\medskip
\begin{center}
\begin{Large}\fontfamily{cmss}
\fontsize{17pt}{27pt}
\selectfont
\textbf{ Finite Dimensional KP $\tau$-functions}
\break \smallskip
I. Finite Grassmannians
\end{Large}

\bigskip \bigskip
\begin{large}  F. Balogh$^{1, 2}$\footnote{e-mail: fbalogh@sissa.it } T. Fonseca$^{1, 3}$\footnote{e-mail:  tiago.dinis.da.fonseca@sapo.pt}  and J. Harnad$^{1, 4}$\footnote {e-mail: harnad@crm.umontreal.ca \\
Work of J.H. supported by the Natural Sciences and Engineering Research Council of Canada (NSERC) and the Fonds Qu\'ebecois de la recherche sur la nature et les technologies (FQRNT). }
 \end{large}\\
\bigskip
\begin{small}
$^{1}${\em Centre de recherches math\'ematiques,
Universit\'e de Montr\'eal\\ C.~P.~6128, succ. centre ville, Montr\'eal,
Qu\'ebec, CANADA, H3C 3J7}\\
\smallskip
$^{2}${\em SISSA, Via Bonomea 265, I-34136 Trieste, ITALY}\\
 \smallskip
 $^{3}${\em LAPTh
Laboratoire d'Annecy-le-Vieux de Physique Th\'eorique\\
9, chemin de Bellevue,  F-74941 Annecy-le-Vieux, FRANCE} \\
 \smallskip
$^{4}${\em Department of Mathematics and
Statistics, Concordia University\\ 1455 de Maisonneuve Blvd. West
Montreal, Quebec, CANADA, H3G 1M8} 
\smallskip
\end{small}
\end{center}
\bigskip
\date{\today}
% \begin{center}Compiled \@date\end{center}
%%%%%%%%%%%%%%%%  Abstract  %%%%%%%%%%%%%%%%
\begin{center}{\bf Abstract}
\end{center}
\smallskip

\begin{small}

We study $\tau$-functions of the KP hierarchy in terms of abelian group actions on finite dimensional Grassmannians,
 viewed as  subquotients of the  Hilbert space Grassmannians of Sato, Segal and Wilson.  A  determinantal formula of Gekhtman and Kasman involving exponentials of finite dimensional matrices is shown to follow naturally from such reductions. All  reduced flows of exponential  type generated by matrices with arbitrary nondegenerate Jordan forms, are derived, both in the Grassmannian setting and within the fermionic operator formalism.  A slightly more general determinantal formula involving resolvents of the matrices generating the flow,  valid on the big cell of the Grassmannian,  is also derived. An explicit expression is deduced for the  Pl\"ucker coordinates appearing as coefficients in the Schur function expansion of the $\tau$-function.  
 \end{small}
\bigskip \bigskip

%%%%%%%%%%%%%%%  Section 1.  %%%%%%%%%%%%%%%%%%%
\section{Introduction}

%%%%%%%%%%%%%%%  SubSection 1. 1 %%%%%%%%%%%%%%%%%%%

\subsection{$\tau$-functions and Hilbert space Grassmannians}

In the approach  to the KP integrable  hierarchy developed by  Sato \cite{S, SS} and Segal and Wilson~\cite{SW},  all solutions are expressed in terms of a $\tau$-function $\tau_W({\bf t})$  of  the  infinite set of KP flow parameters  ${\bf t}=(t_1, t_2, \dots)$,
determined uniquely by  the elements  $W \in  \Gr_{\HH_+}(\HH)$ of an infinite dimensional Grassmann manifold.
   These  are closed subspaces $W \ss \HH $ of a separable Hilbert space  $ \HH$  admitting a natural orthogonal splitting 
\be
\label{H_splitting}
   \HH = \HH_+\oplus \HH_-
\ee
into  the direct sum of two semi-infinite subspaces $\HH_{\pm}$. These are obtained by applying a bounded, invertible linear map 
$g \in GL(\HH)$ to  $\HH_+$
\be
W= g(\HH_+).
\ee  
and are {\em comparable} with the subspace $\HH_+$,  in the sense that the orthogonal projection map $\pi^\perp_+:W \ra \HH_+$ to $\HH_+$  is a Fredholm operator, while  orthogonal projection  $\pi^\perp_-:W \ra \HH_-$  to $\HH_-$  is compact. 
In~\cite{SW}, $\HH$ is taken as  the space $L^2(S^1)$ of square integrable functions $f(z)$ on the unit circle $|z|=1$ in the complex plane and $\HH_+$ and $\HH_-$ are the subspaces of functions with  only positive or negative Fourier components, respectively, with the orthonormal basis $\{e_i := z^{-i-1}\}_{i\in \Zb}$ consisting of the monomials in $z$. 

The Grassmannian $ \Gr_{\HH_+}(\HH)$  is viewed as a universal phase space, with dynamics defined by the action of an infinite abelian group  
\be
\Gamma_+ =\{ \gamma_+({\bf t}):= e^{\sum_{i=1}^\infty t_i \Lambda^i}\} 
\ee
consisting of  flows  generated by the shifts $\Lambda:e_i \mt e_{i-1}$  of the orthonormal basis elements:   
  \bea
\Gamma_+ \times Gr_{\HH_+}(\HH) &\&\ra Gr_{\HH_+}(\HH) \cr
(\gamma_+({\bf t)} ,W)  &\& \ra W({\bf t}):= \gamma_+({\bf t)} W. 
\eea
The flow parameters  ${\bf t}=(t_1, t_2, \dots)$  are thus additive coordinates on the abelian group $\Gamma_+$.  The element $W \in  \Gr_{\HH_+}(\HH)$  parametrizing the $\tau$-function is the initial point  $W({\bf 0})$ of the $\Gamma_+$ orbit $W({\bf t})$ and $\tau_W({\bf t})$ is defined as the  determinant of the orthogonal projection of $W({\bf t})$  to the subspace $\HH_+$:
\be
\tau_W({\bf t}) := \det \left(\pi^{\perp}: W({\bf t}) \ra \HH_+\right ).
\label{tauWdef}
\ee
relative to a suitably defined, admissible basis.

Conversely, knowing  the $\tau$-function is sufficient to determine $W$,   since its Pl\"ucker coordinates  $\pi_{\lambda, N}(W) $  are just the coefficients in the expansion of  
$\tau_W({\bf t})$ in a basis of Schur functions
\be
\tau_W({\bf t}) = \sum_{\lambda}\pi_{\lambda, N}(W) S_\lambda({\bf t}),
\ee
where  $(\lambda, N)$ denotes a pair consisting of an integer partition $\lambda$ and an integer $N\in \Zn$. The latter is the Fredholm index of the orthogonal projection map  $\pi^\perp_+:W \ra \HH_+$  to $\HH_+$, which determines the connected component of the Grassmannian $Gr_{\HH_+}(\HH)$, and is referred to in~\cite{SW} as the {\it virtual dimension} of $W$. As in
finite dimensions, the   Pl\"ucker coordinates  $\{\pi_{\lambda,N}(W)\}$, are not independent since, being maximal minors of
the matrix of homogeneous coordinates, they must satisfy the quadratic  Pl\"ucker relations which, in this setting, form an infinite set. As shown by Sato (\cite{S, SS}), these are equivalent to the Hirota bilinear differential relations for $\tau_W({\bf t})$, which in turn are equivalent to the equations of the KP hierarchy. 
.

\br {\bf Gauge transformations.}
We recall that an invertible  linear transformation  of the form
 \be
W \mapsto \gamma_-({\bf s})W, \quad \gamma_-({\bf s}):= e^{\sum_{i=1}^\infty s_i z^{-i} }
\label{gaugetransform}
\ee
has the effect of multiplying $\tau_W({\bf t})$ by the linear exponential factor $e^{-\sum_{i=1}^\infty  i s_i t_i}$.
\be
\tau_{\gamma_-({\bf s})W}({\bf t}) = e^{-\sum_{i=1}^\infty  i s_i t_i} \tau_W({\bf t}).
\ee
 Since the KP solutions are uniquely determined by the logarithmic derivatives  of the corresponding Baker-Akhiezer function $ \psi_W(z, {\bf t})$, given by the Sato formula 
\be
\psi_W(z, {\bf t}) = e^{\sum_{i=1}^\infty t_i z^i} { \tau_W ({\bf t} - [z^{-1}]) \over \tau_W({\bf t})}, \quad [z^{-1}] := \left({1\over z}, {1\over 2 z^2},  \dots, {1\over i z^i} , \dots\right), 
\ee 
and the transformation~(\ref{gaugetransform}) just multiplies  $ \psi_W(z, {\bf t})$  by the time independent factor $\gamma_-({\bf s})$, this has no effect upon the solutions. These are therefore referred to as {\it gauge transformations}. 
\er

%%%%%%%%% Subsection 1.2 %%%%%%%%%%%%%%%%%

\subsection{Gekhtman-Kasman finite determinantal formula}

Gekhtman and Kasman~\cite{GK1, GK2} found a very simple finite dimensional determinantal expression for a class of KP $\tau$-functions in which the entries have exponential dependence on the flow parameters. These are determined by  a triplet 
of matrices $(A,B,C)$ in which  $A$ and $C$ are  $n \times N$ with $n<N$ and have maximal rank while $B$ is 
a square  $N \times N$ matrix. The finite determinantal formula
\be
\tau^f_{(A,B,C)}({\bf t} )= \det\left(A e^{\sum_{i=1}^\infty t_i B^i} C^T\right)  
\label{tau_GK}
\ee
is easily shown to satisfy the Hirota bilinear relations of the KP hierarchy, 
provided the simple rank-$1$ condition 
\be
\label{rank1}
\rm{rank}(AB (A^\perp)^T) \leq 1
\ee
is satisfied, where $A^\perp$ is any maximal rank $(N-n) \times N$ matrix whose rows are orthogonal to those of $A$. 
That is, they span the 
\be
k:= N-n
\ee
dimensional orthogonal annihilator of the space spanned by the rows of $A$.
For the $\tau$-function not to vanish at the initial time ${\bf t} = \0b$, we must also require that $A C^T$ be nonsingular.

It will be useful to reformulate the rank-$1$ condition in a slightly different way. It is easy to see that~(\ref{rank1}) holds 
if and only if there exists an $n\times n$ matrix $D$ and two vectors $\vf \in \Cn^n$, $\vg \in \Cn^N$ such that the equation
\be
AB-D^TA = \vf \vg^T
\label{rank1_D}
\ee
is satisfied, i.e., that every row of $AB$ can be expressed as a linear combination of the rows of $A$ and the additional fixed vector $\vg^T$. 

The rank in~(\ref{rank1})  is $1$ provided $\vf\vg^T$ is nonzero and $\vg$ does not belong to the row space of $A$. Otherwise $AB=\tilde{D} A$ for some matrix $\tilde{D}$ and therefore
\be
\tau^f_{(A,B,C)}({\bf t} )= \det\left(A e^{\sum_{i=1}^\infty t_i B^i} C^T\right)  =e^{\sum_{i=1}^\infty t_i \mathrm{tr}(\tilde{D})^i}\det(AC^{T}),
\ee
which is gauge equivalent to a constant.

    A slightly more general class of finite determinantal KP $\tau$-functions of exponential type may be constructed as follows. For three positive integers $l, n, N$ with $l \le n$, $l \le N$, we may again choose the matrices $D\in \Mat{n}{n}$, 
  $B \in \Mat{N}{N}$ and a matrix $A \in \Mat{n}{N}$ satisfying the rank-1 condition (\ref{rank1_D})  for some pair of vectors  ${\vf} \in  \Cn^n$,  ${\vg}\in \Cn^N$. Then for any pair of rank-$l$ matrices 
  $F\in \Mat{l }{n}$,  $C\in \Mat{l}{N}$,  the following $l \times l$ determinant
\be
\tau^f_{(A,B,C,D,F)}({\bf t}):= \det(F e^{-\sum_{i=1}^\infty t_i (D^T)^i} A  e^{\sum_{j=1}^\infty t_jB^j}C^T)
\label{tau_gen_GK}
\ee
 is a KP $\tau$-function. The Gekhtman-Kasman formula  (\ref{tau_GK}), corresponds to the special case where $l=n$ and $F$ is an invertible matrix, within the linear exponential factor gauge term $e^{-\sum_{i=1}^\infty\tr(D^T)^i} $. A  simple direct proof that $\tau^f_{(A,B,C,D,F)}({\bf t})$ satisfies the Hirota bilinear relations if the rank-1 condition eq.~(\ref{rank1_D}) is satisfied is given in the Appendix.
 
   In the next subsection, some well-known examples expressible in the form (\ref{tau_GK}) will be recalled. These include: all polynomial $\tau$-functions, giving rise to rational solutions of the KP hierarchy;  all nondegenerate multisoliton solutions, which generally are of exponential  type; and  all  degenerations  of the latter, in particular those that give rise to solutions that are rational in the $t_1=x$ flow variable, with the locus of poles satisfying Calogero-Moser dynamics (cf.~\cite{Kr1, AMM, W}). In the notation of Segal and Wilson~\cite{SW}, the rational solutions appearing in these examples belong to the sub-Grassmannian $\Gr_0$, while the multisoliton solutions and their degenerations belong to the sub-Grassmannian $\Gr_1$. Their place within the general setting is  indicated in Section 1.4. 
   
   Section 2 gives a review of the fermionic approach to $\tau$-functions.  The general case of finite dimensional reductions leading to solutions of  exponential or quasipolynomial type will be derived in detail  in  Sec. 3,  both within the Grassmannian and the fermionic operator formalism.
Sec. 4 gives a solution of the ``inverse problem''; i.e.,  a reconstruction of the element $W(B,C,D)$
corresponding to any set $(A,B,C,D)$ satisfying the rank-1 condition. These are viewed as a specialization of a more general
 class of finite dimensional $\tau$-functions of exponential type, belonging to the big cell. The Pl\"ucker coordinates are explicitly determined for this general class, thereby determining the expansion of the $\tau$-function in a basis of Schur functions.

%%%%%%%%%%%%%%%  SubSection 1. 3 %%%%%%%%%%%%%%%%%%%

\subsection{Examples}

Henceforth, $ \Id{n}$ denotes the $n \times n$ identity matrix and $\Lambda_n$ the upper triangular shift matrix of size $n\times n$:
\be
\Lambda_n := 
   \pmatrix{ 0 & 1 & 0  & \cdots  & 0\cr
    0 & 0  & 1  & \cdots & 0 \cr
     \vdots  &   \vdots &  \vdots  & \ddots & \vdots \cr
       0 &  0  & 0 & \cdots & 1 \cr
         0 &  0   & 0 & \cdots & 0 }  \in \Mat{ n }{n}.
\ee
\begin{example}{\bf Rational solutions.}
\label{ex:rational}
\end{example}

 In formula~(\ref{tau_GK}), choose the following expressions for the matrices $A$ and $B$
\be
  A = \pmatrix{ \Id{n} & \0b }  \in \Mat{n }{N},  \qquad B =  \Lambda_N  \in \Mat{ N  }{N}, 
\ee
where $\0b$ denotes the $n \times k$ matrix whose entries are all $0$'s.  A basis for the orthogonal annihilator of the $n$-dimensional space spanned by the rows of $A$ is given by the columns of the $N \times k$ matrix
      \be
      \left(A^\perp\right)^T := \pmatrix{ \0b \cr \Id{k}}.
    \ee
We have
    \be
     AB \left(A^\perp\right)^T=
   \pmatrix{ 0 & 0  &\cdots & 0 \cr
    0 & 0  &  \cdots & 0  \cr
      \vdots &  \vdots  & \ddots  & \vdots  \cr
             1 &  0  &   \cdots & 0},
 \ee
 so the rank-$1$ condition~(\ref{rank1}) is satisfied. 
 In the version (\ref{rank1_D}) of the rank-$1$ condition, we have 
 \be
 D^T = \Lambda_n   \in \Mat{n }{n}, \quad  {\vfe}_a = \delta_{a,n}, \quad {\vge}_b = \delta_{b,n+1}, \quad a =1, \dots n,\ b = 1, \dots N.
 \ee
 The matrix $C$ can be any  $n\times N$ matrix of maximal rank, which may be viewed as the 
 homogeneous coordinates of an element $[C] \in \Gr_n(\Cn^N)$  of the Grassmannian of $n$-dimensional subspaces of $\Cn^N$. 
   For any partition $\lambda$ whose Young diagram fits into that of the rectangular partition $(k)^n$ we let $C_\lambda$ denote the $n \times n$ minor whose $i$th column is the $(\lambda_i -i +n+1)$th column of $C$. The corresponding  Pl\"ucker coordinate $\pi_\lambda(C)$ of $[C] \in \Gr_n(\Cn^N)$ is then
   \be
\pi_\lambda(C) = \det(C_\lambda). 
\ee
It follows from the Cauchy-Binet identity that the expansion of the $\tau$-function~(\ref{tau_GK}) 
in a basis of Schur functions $S_\lambda({\bf t})$ is given by
\be
\tau^f_{(A, \Lambda_N,C)}({\bf t}) = \sum_{\lambda \ss (k)^n} \pi_\lambda(C) S_\lambda({\bf t}).
\ee
This is the general form of KP $\tau$-functions that have a polynomial  dependence on all the
KP  flow parameters, which give rise to   solutions of the hierarchy that are rational in all these
variables.

\begin{example}{\bf KP solitons.}
\label{ex:soliton}
\end{example}

Now choose $B$ to be the diagonal matrix 
\be
\label{diagB}
B=B( \beta)   := \diag\{\beta_i\}_{i=1}^{N} \in \Mat{N }{N}
\ee
with distinct eigenvalues $\{\beta_i\}_{i=1\cdots N}$, and  $A$ to be the truncated 
$n \times N$ Vandermonde matrix
\be
\label{VdM_nN}
A_V (\beta):= V_{n, N} (\beta) =  \pmatrix{  \beta_1^{n-1} &  \beta_2^{n-1}  & \cdots &  \beta_N^{n-1} \cr
      \beta_1^{n-2}  &   \beta_2^{n-2}   & \cdots     & \beta_N^{n-2}  \cr
       \vdots  &   \vdots &  \ddots  & \vdots \cr
        \beta_1 &   \beta_2 & \cdots   & \beta_N\cr
        1 &  1   &  \cdots & 1}
\ee
Let
\be
p(z) := \det(z\Id{N}-B) = \prod_{j=1}^N (z -\beta_j)
\ee
 be the characteristic polynomial of $B$.
 It  follows from the Cauchy residue theorem applied to
\be
{1\over 2 \pi i} \oint_\infty {z^j dz \over p(z)}  = 0, \quad  {\rm for}\  j < N-1
      \label{Cauchy_beta_j}
\ee
 that the orthogonal complement of the subspace spanned by the rows of $A_V(\beta)$ is spanned by the rows of the $k \times N$ matrix
 \be
 A_V(\beta)^\perp=  \pmatrix{  \beta_1^{k-1} &  \beta_2^{k-1}  & \cdots &  \beta_N^{k-1} \cr
      \beta_1^{k-2}  &   \beta_2^{k-2}   & \cdots     & \beta_N^{k-2}  \cr
       \vdots  &   \vdots &  \ddots  & \vdots \cr
        \beta_1 &   \beta_2 & \cdots   & \beta_N\cr
        1 &  1   &  \cdots & 1} 
         \pmatrix{ {1\over p'(\beta_1)} & 0 & 0  & \cdots  & 0\cr
    0 & {1\over p'(\beta_2)}  & 0  & \cdots & 0 \cr
      0 &   0  & 0   & \cdots  & 0 \cr
       \vdots  &   \vdots &  \vdots & \ddots & \vdots \cr
         0 &  0   & 0  & \cdots &{1\over  p'(\beta_N)} }.
 \ee

The rank-$1$ condition
\be
A_V(\beta) B (A_V(\beta)^\perp)^T  = 
\pmatrix {1 & 0  & \cdots  & 0\cr
      0 &   0  &  \cdots  & 0 \cr
       \vdots  &   \vdots &  \ddots & \vdots \cr
         0 &  0     & \cdots &0}
  \label{Abeta}
         \ee
follows from the Cauchy residue theorem applied to
\be
{1\over 2 \pi i} \oint_\infty {z^{N-1} dz \over p(z)}  = 1,
\ee
together with~(\ref{Cauchy_beta_j}).
The version (\ref{rank1_D}) of the rank-$1$ condition is then satisfied, with
\be
D    = \Lambda_n,  \quad 
 \vf := \pmatrix{ 1 \cr 0 \cr \vdots \cr 0} \in \Cn^n,  \quad 
 \vg:=  \pmatrix{ \beta_1^n \cr \beta_2^n \cr \vdots \cr \beta_N^n} \in \Cn^N.
 \label{diagD}
\ee

The resulting expression
\bea
\tau^f_{(A_V(\beta), B(\beta), C)}({\bf t}) &\& = \det \left( A_V(\beta) e^{\sum_{i=1}^\infty t_i B^i(\beta)}C^T\right) 
\label{tauAbetaBbetaC}\\
&\&= \sum_{\lambda \ss (k)^n}\pi_\lambda(C) e^{T_\lambda(\beta, {\bf t})},
\label{exponential_sum}
\eea
where
\bea
T_\lambda(\beta, {\bf t}) &\&:= \sum_{i=1}^\infty t_i \sum_{j=1}^n \beta_{\ell_j}^i \\
\ell_j &\&  := \lambda_j - j + n + 1, 
\eea
is  the $\tau$-function for the general  rank $n$, $N$-soliton  solution of the KP hierarchy. 
The second equality (\ref{exponential_sum}) follows  from the Cauchy-Binet theorem
applied to the product of the $n \times N$ and $N\times n$ matrices appearing in the determinant. 
The $\tau$-function is real for real  flow parameters ${\bf t} =(t_1, t_2 \dots)$  if the $\beta_i$'s  and the  matrix $C \in \Mat{n }{N}$ are real.
It is nonvanishing, giving rise to nonsingular solutions,  provided the $\beta_i$'s are strictly decreasing and $C$  has only nonnegative Pl\"ucker coordinates, i.e., provided the space spanned by the rows of $C$ belongs to the nonnegative Grassmannian  $\Gr^+_n(\Rn^N)$~\cite{Ko, KoW1, KoW2}.

Another variant of the determinantal form of the above solution may be obtained by choosing $A$ as the Cauchy matrix:
\be
  \label{Abetadelta}
A= A^0_C(\beta, \delta)   := 
\pmatrix {{1\over \beta_1 - \delta_1}  & {1\over \beta_2 - \delta_1} &  \cdots  & {1\over \beta_N - \delta_1}\cr
    {1\over \beta_1 - \delta_2} &{1\over \beta_2 - \delta_2}&  \cdots & {1\over \beta_N - \delta_2} \cr
            \vdots  &   \vdots &  \ddots  & \vdots  \cr 
    {1\over \beta_1 - \delta_n} &  {1\over \beta_2 - \delta_n}   & \cdots &{1\over \beta_N - \delta_n}}.
 \ee
 In this case, the orthogonal annihilator is spanned by the rows of
  \be
  (A^0_C (\beta,\delta))^\perp
  =
  \pmatrix {{r(\beta_1)\over (\beta_1 - \delta_{n+1}) p'(\beta_1)}  & {r(\beta_2)\over (\beta_2 - \delta_{n+1}) p'(\beta_2)}&  \cdots  &{r(\beta_N)\over (\beta_N - \delta_{n+1}) p'(\beta_N)}\cr
    {r(\beta_1)\over (\beta_1 - \delta_{n+2}) p'(\beta_1)} &{r(\beta_2)\over (\beta_2 - \delta_{n+2}) p'(\beta_2)}&  \cdots &{r(\beta_N)\over (\beta_N - \delta_{n+2}) p'(\beta_N)} \cr
            \vdots  &   \vdots &  \ddots  & \vdots  \cr 
    {r(\beta_1)\over (\beta_1 - \delta_{N}) p'(\beta_1)} & {r(\beta_2)\over (\beta_2 - \delta_N )  p'(\beta_2)}& \cdots &{r(\beta_N)\over (\beta_N - \delta_N) p'(\beta_N)}}.
  \label{Abetadeltaperp}
  \ee
  where $\{\delta_{n+1}, \dots \delta_N\}$ is any further set of distinct constants,
  unequal to the previous $\delta_i$'s or  $\beta_i$'s and
  \be
  r(z) := \prod_{i=1}^N (z-\delta_i).
  \ee
This follows from Cauchy's theorem applied to
\be
{1\over 2 \pi i} \oint_\infty {r(z) dz \over {(z-\delta_j) (z-\delta_k)p(z)}}  = 0,
 \quad  {\rm for}\  1\le j \le n < k \le N.
      \label{Cauchy_delta_jk}
\ee
It also follows from Cauchy's theorem applied to
\be
{1\over 2 \pi i} \oint_\infty { z\, r(z) dz \over {(z-\delta_j) (z-\delta_k)p(z)}}  = 1
      \label{Cauchy_delta_jk}
      \ee
      that
\be
A^0_C(\beta,\delta) B((A^0_C(\beta,\delta))^\perp)^T  = 
\pmatrix {1 & 1 & 1  & \cdots  & 1\cr
    1 &1& 1  & \cdots & 1 \cr
      1 &   1  & 1   & \cdots  & 1 \cr
       \vdots  &   \vdots &  \vdots & \ddots & \vdots \cr
         1 &  1   & 1  & \cdots &1}
  \label{Adelta_beta}
         \ee
and hence has rank $1$.

The version (\ref{rank1_D}) of the rank-$1$ condition 
\be
A^0_C(\beta,\delta)B(\beta) - D^TA^0_C(\beta,\delta) = \vf \vg^T
\ee
is satisfied with
\bea
\label{diagD}
D = D (\delta)  &\&  := \diag\{\delta_i\}_{i=1}^{n} \in \Mat{n }{n} 
         \eea
         and
         \be
 \vf := \pmatrix{ 1 \cr \vdots \cr 1} \in \Cn^n,  \quad \vg :=  \pmatrix{ 1 \cr \vdots \cr 1} \in \Cn^N.
\ee
  In fact, the corresponding $\tau$-function
  \be
\tau^f_{(A_C^0(\beta,\delta), B(\beta), C)}({\bf t})= \det \left( A^0_C(\beta,\delta) 
e^{\sum_{i=1}^\infty t_i (B(\beta))^i} C^T\right)
\label{tauAdelta_betaBbetaC}
\ee
is of the same class as~(\ref{tauAbetaBbetaC}), and differs from it only by a slight modification of the choice of the matrix $C$. It
therefore simply represents another parametrization of the multisoliton solutions, whatever the choice of the constants $\{\delta_i\}_{i=1, \dots n}$.

To see this, define the  $n\times n$ matrix
\be
K_{ab} (\delta) = {(\delta_b)^{n-a} \over r'(\delta_b)}, \quad 1 \le a, b \le n.
\ee
By evaluating the integral
\be
(\beta_j)^a = {1\over 2\pi i}\oint_{|z|=1}  {z^{n-a} \over z-\beta_j} dz
\ee
using the Lagrange interpolation formula
\be
z^{n-a} = \sum_{b=1}^n {r(z) \over z-\delta_j} {(\delta_b)^{n-a}\over r'(\delta_b)}, 
\ee
we obtain the matrix product identity
\be 
A_V(\beta) = K(\delta)A^0_C(\beta,\delta) r(B(\beta)),
\label{KAR}
\ee

It follows from (\ref{KAR}) that the $\tau$-function of eq.~(\ref{tauAbetaBbetaC}) can be equivalently written as
\be
\tau^f_{(A_V(\beta), B(\beta), C)}({\bf t}) =\kappa(\delta) \tau^f_{(A_C^0(\beta,\delta), B, r(B(\beta))^{T} C)}({\bf t})
\ee
where
\be
\kappa(\delta) :=  \det(K(\delta)), 
\ee
since $r(B(\beta))$ commutes with $B(\beta)$.  Thus $\tau^f_{(A_V(\beta), B(\beta), C)}({\bf t})$
is just a multiple of  $\tau^f_{(A_C^0(\beta,\delta), B, C)}({\bf t})$ with
$C$ replaced by $r(B)^{T} C$.
Since the choice of $C$ is arbitrary, the two sets of $\tau$-functions~(\ref{tauAbetaBbetaC}) and~(\ref{tauAdelta_betaBbetaC}) coincide. To assure the reality and positivity condition however, it is necessary that the $\delta_i$'s be real, and that all the entries of the diagonal matrix $r(B)$ be of the same sign. This will be satisfied if the $\delta_i$'s are chosen to be less than all the $\beta_i$'s:
\be
\delta_i < \beta_j , \quad  i= 1, \dots n, \ j=1, \dots N.
\ee
\break 
\begin{example}{\bf Generic case:  degeneration of KP solitons.}\label{ex:KPdegenerations} 
\end{example}
 \nobreak

More generally,  both $B\in \Mat{N }{N}$ and $D \in \Mat{n }{n}$  may have any Jordan structure. 
Without loss of generality, we may choose them to be in  standard upper triangular Jordan normal form,  
with distinct  eigenvalues $\{\beta_j\}_{j=1, \dots , M}$ for $B$  
 \be
\label{B_jordan}
 B=
 \pmatrix{ J_{N_1}(\beta_1) & 0 & 0 & 0 \cr 0 & J_{N_2}(\beta_2) & 0 & 0 \cr 0 & 0 & \ddots & 0 \cr 0 & 0 & 0 & J_{N_M}(\beta_M)}
\ee
where
\be 
J_{N_j} (\beta_j) = \beta_j\, \Id{N_j} + \Lambda_{N_j}, \quad  j=1, \dots, M 
\ee 
denotes a Jordan block  of dimension  $\{N_j\}$,  and  eigenvalue $\beta_j$ and $D$ similarly is of the form
\be
\label{D_jordan}
 D =
 \pmatrix{ J_{n_1}(\delta_1) & 0 & 0 & 0 \cr 0 & J_{n_2}(\delta_2) & 0 & 0 \cr 0 & 0 & \ddots & 0 \cr 0 & 0 & 0 & J_{n_m}(\delta_m)}
\ee
with Jordan blocks of dimension  $\{n_j\}$ and distinct eigenvalues $\{\delta_i\}_{i=1, \dots , m}$, also chosen to be distinct from the  $\beta_j$'s. 
The equation
 \be
 AB - D^T A = \vf\vg^T
\label{rank1D}
  \ee
 then has a unique solution for any given pair of nonvanishing vectors $\vf\in \Cn^n$, $\vg\in \Cn^N$.

We may always multiply on the right  by an element of the stabilizer $G_B \ss GL(N)$ of $B$ under conjugation or
on the left by an element of the stabilizer $G_D \ss GL(n)$ of $D^T$ and obtain a new
solution that gives on equivalent class of $\tau$-functions.
   For such general $B$ and $D$,  the solution $A(B,D)$ to the rank-$1$ equation (\ref{rank1D}) for a suitable choice
  of ${\vf}$ and ${\vg}$, is given in  eq.~(\ref{ABD})) of  in  Sec. 3.2 and  Proposition \ref{prop_AB_rank1}, with $r(z)$ and $p(z)$ replaced by the characteristic polynomials $r_D(z)$ and $r_B(z)$ of the matrices $D$ and $B$ respectively. 
Denoting by $A(B) := A(B, \Lambda_n)$  the special case when $D$ is chosen as the shift matrix $\Lambda_n$, it follows, as in the above special case, that  the $\tau$-function 
  \be
  \tau^f_{(A(B,D), B, C)}({\bf t))} = \det( A(B,D) e^{\sum_{i=1}^\infty t_i B^i} C^T)
  \ee
 determined by the triple $(A(B,D), B, C)$, as given by Theorem \ref{theorem_tau_BCD}, coincides with $\tau^f_{(A(B), B, C)}({\bf t})$ within a multiplicative constant. Therefore, the choice of $D$  in the form of the  rank-$1$ condition~(\ref{rank1D}) does not affect the resulting class of solutions.

The next example is a special case of  nondiagonal $B$,  in which $N=2n$, and 
$B$ consists of $n$ distinct $2\times 2$ Jordan blocks, with a special choice of $C$, 
which gives rise to pole dynamics of the Calogero-Moser type.

\begin{example}{ Calogero-Moser pole dynamics.~(\cite{Kr1, AMM, W})}
\label{ex:CM}
\end{example}

  Choose $B$ to be a $2n \times 2n$ matrix of the form
\be
B = B_Z:= \pmatrix{ Z  & \Id{n} \cr
\0b & Z}
\ee
where $Z$ is the diagonal $n\times n$ matrix
\be
Z=\diag\{\beta_i\}_{i=1}^{n}
\ee
with distinct eigenvalues $\{\beta_i\}_{i=1\cdots N}$.
For $A$, choose the modified truncated Vandermonde matrix
\be
A_{V'(}\beta):= \pmatrix { V_{n, n} (\beta) &  V'_{n, n} (\beta)}
\ee
where $V_{n,n}(\beta)$ is defined as in (\ref{VdM_nN}) and
\be
V'_{n, n} (\beta) :=  \pmatrix{ (n-1) \beta_1^{n-2} &  (n-1)\beta_2^{n-2}  & \cdots &  (n-1)\beta_n^{n-2} \cr
 (n-2)  \beta_1^{n-3}  &   (n-2)  \beta_2^{n-3}   & \cdots     &  (n-2) \beta_n^{n-3}  \cr
       \vdots  &   \vdots &  \ddots  & \vdots \cr
       1 &   1 & \cdots   &1 \cr
        0 &  0   &  \cdots & 0},
\ee
  and take $C$ to be of the special form
  \be
  C^{T} = C^{T}_\Xi :=\pmatrix{ \Id{n}\cr \Xi},
  \ee
  where
 \be
  \Xi = \diag\{\xi_i\}_{i=1}^{n}.
 \ee
 The rank-$1$ condition is easily verified by applying the Cauchy theorem to
 \be
{1\over 2 \pi i} \oint_\infty {z^j dz \over p^2(z)} = \delta_{j,2n}
 \ee
where
\be
\det\left( z\,\Id{n} - B_Z\right) = p^2(z) := \prod_{i=1}^n (z-\beta_i)^2.
\ee
 
The resulting $\tau$-function is of the form
\be
\tau^f_{(A_{V'}(\beta), B_Z, C_\Xi)}({\bf t} )= e^{\sum_{i=1}^\infty t_i Z^i} \det (V_{n,n}(\beta)) \det \left( X_0 + \sum_{i=1}^\infty i t_i Z^{i-1} \Xi \right)
\ee
where
\be
X_0 = \Id{n} + V^{-1}_{n, n} (\beta) V'_{n, n} (\beta) \Xi, 
\ee
which is gauge equivalent to the $\tau$-function for rational solutions of the KP hierarchy in which the pole dynamics are determined by the Calogero-Moser $n$-particle system~\cite{Kr1, AMM, W}. More general solutions in this class may be obtained by allowing the matrix $Z$ to have general Jordan normal form, and $\Xi$  to be an element of its centralizer.

%%%%%%%%%%%%%%%  SubSection 1.4  %%%%%%%%%%%%%%%%%%%

\subsection{Finite dimensional reductions of Grassmannians}

The reduction to finite dimensional  systems may be viewed as a two-step process: first the identification of a fixed subspace $W_2 \ss \HH$, invariant under the flows,  that contains $W$ as a finite codimensional subspace. Second,  a quotient by another  fixed finite 
codimensional subspace $W_1 \ss W_2$, also invariant under the flows, that is contained in $W$. In the case where these subspaces are chosen as
\be
W_1 = r(z) \HH_+, \quad W_2 =  {r(z) \over p(z)} \HH_+
\ee
for a pair of polynomials $r(z)$,  $p(z)$ of degrees $n$ and $N$,  respectively, with the roots of both inside the unit circle, we obtain (within gauge equivalence) precisely the Grassmannian $\Gr_1$ of~\cite{SW}.
The corresponding pair of matrices $B$ and $D$ are those whose eigenvalues coincide with the roots of $p(z)$ and $r(z)$, respectively,  with Jordan blocks of dimension equal to the degree of these roots. This determines, up to conjugation,  a unique pair of regular elements, 
$B\in \grg\grl(N)$, $D\in \grg\grl(n)$  whose characteristic polynomials are $p(z)$ and $r(z)$, respectively. Within gauge equivalence, there is no loss of generality in assuming that the roots of $r(z)$ and those of $p(z)$ are mutually distinct. 

The finite dimensional reduction may be viewed as a subquotient. Projecting $W \ra W/{W_1}$ gives an element of the finite dimensional Grassmannian $\Gr_{n}(W_2/W_1)$ of $n$-dimensional subspaces of $W_2/W_1$, and $W_2/W_1$ can be identified with $\Cn^N$ through the choice of a suitable basis. The resulting flows can  be expressed in terms of the Pl\"ucker coordinates of $W/W_1 \ss W_2/W_1$, and the corresponding KP $\tau$-function becomes a finite determinant having linear exponential or quasi polynomial dependence on the flow variables.  The generator of the reduced flow is  a matrix  $B$ that may have any nondegenerate Jordan normal form, which is uniquely determined by the choice of basis for $W_2/W_1$.

In particular, $B$ could be nilpotent,  consisting of a single Jordan block with zero eigenvalue;  i.e.,  the $N \times N$ ``shift'' matrix  $\Lambda_N$, whose characteristic polynomial is the monomial $p(z)= z^N$. This naturally gives rise to the polynomial $\tau$-functions of Example \ref{ex:rational} above.  Alternatively, choosing $p(z)$ as the monic polynomial with distinct roots  $\{\beta_i\}_{i=1, \dots N}$ results in  flows generated by the finite nondegenerate diagonal matrix with these eigenvalues, as in Example \ref{ex:soliton}. 
The various other cases can be obtained by allowing multiple zeros in $p(z)$, which give rise to reduced flow generators $B$ having all possible Jordan normal forms with distinct eigenvalues.  Special cases of such  degenerations of exponential or trigonometric  multisoliton solutions may be used to embed certain finite dimensional integrable systems, such as the Calogero-Moser system of Example \ref{ex:CM}, as the dynamics of poles of rational solutions of the KP hierarchy. 

  The purpose of this paper is to provide a geometrical construction of such finite dimensional $\tau$-functions through the process of reduction from the infinite case.  We use subquotients to define families of embeddings of finite Grassmannians into infinite ones and deduce thereby the triplets  $(A,B,C)$. The matrix $A$ is determined by the choice of the fixed subspace $W_1$,  and the basis for a complement $W_1^c\ss W_2$ of $W_1$ in $W_2$. The latter also
 determines the generating matrix $B$ of the reduced flows. The element $W$, viewed as an extension of $W_1$ by a subspace of $W_1^c$ determines the matrix $C$, and conversely.  The projection  $W \ra W/W_1 \ss W_2/W_1$  allows us to identify $C$ as the homogeneous coordinates of an initial point in the finite Grassmannian $\Gr_n(W_2/W_1)$ which is identified, through the choice of  basis for $W_1^c$, as a subspace of $\Cn^N$. 
\br
There are other instances of $\tau$-functions expressible as finite dimensional determinants, in which there is no known interpretation in terms of finite dimensional Grassmannians. For instance, the partition function in random matrix models, in which the underlying conjugation invariant  measure is subject to linear exponential deformations, is known to be a KP $\tau$-function that admits a finite dimensional determinantal representation in terms of the Hankel matrix formed from the moments.  However, this does not seem to fit within the finite dimensional reduction framework discussed here, since the dependence upon the flow parameters  is not exponential or quasipolynomial.

Other cases, such as the solutions of the KP-hierarchy expressible in terms of Riemann $\theta$-functions on the Jacobi varieties of an algebraic curve, also involve a reduction to a finite number of degrees of freedom~\cite{Kr2, Du}.  However, the resulting $\tau$-function is not  known to be expressible as a finite dimensional determinant.
\er

   %%%%%%%%%%%%%%% Section 2. %%%%%%%%%%%%%%%%%%%
\section{KP $\tau$-functions}

   %%%%%%%%%%%%%%% Section 2. 1%%%%%%%%%%%%%%%%%%%

\subsection{Grassmannians and fermionic Fock space}

Following Segal and Wilson~\cite{SW}, the model for the Hilbert space $\HH$ we use is the space $L^2(S^1)$ of square integrable functions on the unit circle $S^1=\{z \in \Cn, |z|=1\} $ in the complex $z$-plane. This splits into the direct sum
\be
\HH = \HH_+ \oplus \HH_-
\ee
of subspaces $\HH_+$ and $\HH_-$, consisting, respectively, of  functions admitting a holomorphic continuation
to the interior and exterior of the unit circle, with the latter vanishing at $\infty$. These may be viewed as 
completions of the span of the positive and negative monomials in $z$ 
\be
\HH_+  =\overline{ \span \{z^i\}_{i \in \Nn}},  \quad \HH_-  = 
\overline{\span \{z^{-i}\}_{i \in \Nn^+}}.
\ee
   For consistency with other conventions, it is convenient to label the monomial basis as
   \be
   e_i := z^{-i-1}, \quad i \in \Zn.
    \label{orthobasis}
   \ee
   Then  $\HH_+ $ and $\HH_- $ are mutually orthogonal with respect to
the complex inner product $(\ , \ )$  in which these are orthonormal
   \be
   (e_i, e_j) = \delta_{ij}.
   \ee

The elements of the Grassmannian $\Gr_{\HH_+}(\HH)$ are subspaces $W\ss \HH$ that are comparable with $\HH_+$,  in the sense that  orthogonal projection to $\HH_+$
\be 
\pi^{\perp}_+: W \ra \HH_+
\ee
along $\HH_-$ is a Fredholm operator, while projection to $\HH_-$ along $\HH_+$ 
\be 
\pi^{\perp}_-: W \ra \HH_-
\ee
is compact. The Fredholm index $N$ of the projection map $\pi^{\perp}_+: W \ra \HH_+$ is called the ``virtual dimension'' of $W$. The subspace
\be
\HH_+^N := z^{-N}   \HH_+ \ss \HH,
\ee
in particular, has virtual dimension $N$. The connected components of $\Gr_{\HH_+}(\HH)$ consist of those $W\in \Gr_{\HH_+}(\HH)$ with virtual dimension $N \in \Zn$. These  may be viewed as the orbit $\Gr_{\HH_+^N}(\HH)$ of $\HH_+^N$ under the identity component  $GL_0(\HH)$ of the restricted infinite dimensional Lie group $GL_{res}(\HH)$ of invertible linear transformations of $\HH$ having a well-defined determinant and preserving the properties defining the elements of  $\Gr_{\HH_+}(\HH)$. (See~\cite{SW} for  more detailed definitions.)

%%%%%%%%%%%%%%%  SubSection 2.2 %%%%%%%%%%%%%%%%%%%

\subsection{Fermionic Fock space and the Pl\"ucker embedding}

The fermionic Fock space $\FF$  is defined as the semi-infinite exterior space
    \be
   \FF = \bigwedge \HH
   \ee
 spanned by an orthonormal basis   $\{\vert \lambda; N\rangle\}$ consisting of semi-infinite wedge products 
   \be
   |\lambda ; N\rangle := e_{l_1} \wedge e_{l_2} \wedge  \cdots,
   \ee
 where $\{l_1, l_2, \dots \}$ is a strictly decreasing sequence of integers  $l_1> l_2 > \cdots$ , eventually stabilizing on  a consecutive sequence of decreasing integers. The partition 
    \be
   \lambda :=\{ \lambda_1 \ge \lambda_2 \ge \lambda_{\ell(\lambda)} >0\}
      \label{partition}
   \ee
    of length $\ell(\lambda)$, is related to the sequence by
   \be
   l_i := \lambda_i -i + N
   \ee
   (with the convention that $\lambda_i =0$ for $i> \ell(\lambda)$).
  The integer  $N$ is  the largest one below the ${1\over 2}$-integer point $\nu \in \Zn + {1\over 2}$ on the real line such that, 
   if all integer sites $\{l_i\}$ are viewed as ``occupied'' and all others as unoccupied, there are 
   as many unoccupied sites to the left of $\nu$ as there are occupied sites to the right. 
  The fermionic Fock space $\FF$ thus admits a decomposition
   \be
   \FF =\bigoplus _{N\in \Zn} \ \FF_N
   \ee
  as an orthogonal direct sum of the subspaces $\FF_N$ spanned by basis elements of charge $N$.
      The basis element
   \be
   |0; N\rangle = e_{N-1} \wedge e_{N-2} \wedge \cdots := \vert N \rangle \in \FF_N
   \ee
   is referred to as the charge $N$ vacuum state, and denoted simply as $| N\rangle$.
    (The reason for the seemingly reversed sign convention in (\ref{orthobasis}) is that,
under the Pl\"ucker map, the element $\HH_+\in \Gr_{\HH_+}(\HH)$ should correspond to the vacuum
element $|0\rangle$, which is the ``Dirac sea'', in which all the {\it negative} integer lattice sites
are occupied.)

  As in finite dimensions, the Grassmannian $\Gr_{\HH_+}(\HH)$, and each of its connected components  $\Gr_{\HH^N_+}(\HH)$, may be viewed as infinite dimensional analogs of algebraic varieties,   since they can be embedded into the projectivization $\Pr(\FF)$  by the Pl\"ucker map
       \bea
\grP\grl: \Gr_{\HH_+}(\HH) &\& \ra \Pr(\FF) \cr
 \grP\grl: \span \{w_1, w_2, \dots \}&\& \mapsto [w_1 \wedge w_2 \wedge \cdots ],
 \label{Pluckermap}
 \eea
  where $[ \cdots ]$ denotes the projective class, and the image $\grP\grl \left( \Gr_{\HH_+}(\HH)\right)\ss \Pr(\FF)$ is
  the intersection of an infinite number of quadrics, defined by the Pl\"ucker relations.
 It follows from~(\ref{Pluckermap}) that the image $\grP\grl \left( \Gr_{\HH_+}(\HH)\right)\ss \Pr(\FF)$
  consists of all decomposable elements in $\Pr(\FF)$,  while the image of the virtual dimension $N$ component $Gr_{\HH^N_+}(\HH)$
is  in $\Pr(\FF_N)$.
 In particular, the  image   $\grP\grl(\HH_+)$ of  the element $\HH_+ \in \Gr_{\HH_+}(\HH_+)$  
 is (the projectivization of) the vacuum  element $|0\rangle$. 
 \be
 \grP\grl: \HH_+ \mapsto [|0\rangle]
 \ee
 
  From the definition of the Pl\"ucker map and the scalar product on $\FF$, it follows that  the Pl\"ucker coordinates 
        \be
        \pi_{\lambda,N}(W) := \langle \lambda;N \vert \grP\grl(W)\rangle
        \ee
are determinants of the semi-infinite matrices that appear as maximal minors of the matrix of homogeneous coordinates  of $W \in Gr_{\HH_+}(\HH)$  relative to the given orthonormal basis.
In what follows, it will be sufficient to consider only elements $W\in\Gr_{\HH_+}(\HH)$ that have virtual dimension $0$ and hence, unless otherwise needed, the index $N$ labelling the Pl\"ucker coordinate will be understood to be $0$, the basis states $|\lambda; 0\rangle$ denoted  simply as $|\lambda \rangle$ and the Pl\"ucker coordinates  as
 \be
        \pi_{\lambda}(W) := \langle \lambda \vert \grP\grl(W)\rangle
        \ee
The determinantal formula~(\ref{tauWdef}) defining the $\tau$-function may be interpreted as the Pl\"ucker coordinate  of the element $W({\bf t})$  corresponding to the trivial partition, 
\be
\tau_W({\bf t}) = \pi_{0} (W({\bf t})).
\ee

Relative to the  basis $\{e_i\}_{i\in\Zn}$, we have the standard cellular decomposition, in which the ``big cell'' consists of all elements $W\in\Gr_{\HH_+}(\HH)$ that can be represented as the graph of a linear map $\AA: \HH_+ \ra \HH_-$. A basis for such an element may be chosen to consist of elements of the form
\be
w_i := e_{-i-1} + \sum_{j=0}^\infty \AA_{ij}e_j, \quad i \in \Nn
\label{affine_coords}
\ee
where  the elements $\{\AA_{ij}\}_{i, j \in \Nn}$ of the semi-infinite matrix $\AA$  are standard affine coordinates on the big cell. It follows from the definition of the Pl\"ucker coordinates
that these coincide, within a sign, with the Pl\"ucker coordinates corresponding to hook partitions $(i+1, (1)^j)$ which, in Frobenius notation are denoted $(i | j)$
\be
\AA_{ij}= (-1)^b \pi_{(i|j)} (W)
\ee
More generally, denoting a partition $\lambda$ in Frobenius notation as $(a_1, \cdots , a_r | b_1, \cdots , b_r)$, where $a_i$ is the number of boxes to the right of the $(i,i)$ diagonal element of the Young diagram and $b_i$ the number of elements beneath it, 
the Pl\"ucker coordinate $\pi_{(a_1, a_2, \cdots a_k | b_1, b_2, \cdots b_k)}$ may be expressed, on the  coordinate neighborhood of the big cell, in terms of those for the hook partitions through a generalized Giambelli formula~\cite{HE}:
\be
\pi_{(a_1, a_2, \cdots a_k | b_1, b_2, \cdots b_k)} = (-1)^{\sum_{i=1}^k b_i} \det(\AA_{a_i, b_j})
\label{gen_giambelli}
\ee
        
 The image  $\grP\grl(W({\bf t}))$  of the $\Gamma_+$ orbit $W({\bf t})$ may be simply expressed  in terms of fermionic creation and annihilation operators $\{\psi_i, \psi^\dag_i\}_{i\in \Zn}$  defined, respectively,  as exterior products with the basis elements $\{e_i\}$, and interior products with the dual basis element $\{\tilde{e}^i\}$.
   \be
   \psi_i := e_i \wedge, \qquad \psi^\dag_i :=\tilde{e}^i\lrcorner, \qquad  i \in \Zn.
   \label{psi_i_psidag_i}
   \ee
      These satisfy the usual anticommmutation relations
\be
[\psi_i, \psi_j]_+ = [\psi^\dag_i, \psi^\dag_j]_+ =0, \qquad  [\psi_i, \psi^\dag_j]_+ = \delta_{ij}.
\label{anticomm}
\ee
				and span the subspace of linear elements of the  Clifford algebra associated to the  group of orthogonal transformations $O(\HH+\HH^*, Q)$ preserving the natural quadratic form 
\be
Q(X, \nu)) = 2\nu(X) , \quad X \in \HH, \quad \nu \in \HH^*
\ee
on the sum $\HH + \HH^*$ of the underlying Hilbert space and its analytic dual. We also have the fermionic field operator  $\psi(z)$, and its dual $ \psi^\dag(z)$, 
\be
 \psi(z) := \sum_{i\in \Zn} \psi_i z^i, \quad  \psi^\dag(z) := \sum_{i\in \Zn} \psi^\dag_i z^{-i-1}, 
\ee
which may be viewed as generating functions for the $\psi_i$'s and $\psi_i^\dag$'s.

The subgroup $GL(\HH) \ss O(\HH+\HH^*, Q)$ of general linear transformations $GL(\HH) $, and its abelian subgroup $\Gamma_+ \ss GL(\HH)$, 
generating the commuting KP flows  act naturally on the exterior space through the fermionic representation. 
\be
g:= e^\xi \mapsto \hat{g} := e^{\sum_{i,j \in \Zn}\xi _{ij}\psi_i \psi^\dag_j}
\ee
where $\xi_{ij}$ are the matrix components of the Lie algebra element $\AA\in \End({\HH})$ in the $\{e_i\}$ basis. 
In this notation the fermionic representation of the elements $\gamma_+({\bf t}) \in \Gamma_+$ defining the KP flows is
\be
\hat{\gamma}_+({\bf t}) = e^{\sum_{i\in \Zn} t_i J_i} 
\ee
where 
\be 
J_i := \sum_{k\in \Zn} \psi_k \psi^\dag_{k+i}, \quad i \in \Nn^+.
\ee
are the generators of the ``shift'' flows in the fermionic representation (which are Fourier components of
 the {\it current} operator).
 
It follows that the KP $\tau$-function $\tau_W({\bf t})$ may equivalently be represented as the vacuum state expectation value of a product of such group elements
\be
   \tau_W({\bf t}, N)  =\langle N |\hat{ \gamma}_+({\bf t}) \hat{g} |N \rangle,
   \label{VEV_tau}
\ee
where $g \in GL_0(\HH)$ is any element that takes $\HH_+^N$ into $W$
    \be
    W=g(\HH_+),
    \ee
    and  $N$ is the Fredholm index of the projection map $\pi^\perp_+:W \ra \HH_+$. 
Eq.~(\ref{VEV_tau}) may be understood as defining  the $| N \rangle$ component of $\grP\grl(W({\bf t}))$ (which is nonzero only if $W$ has virtual dimension $N$); i.e., the Pl\"ucker coordinate $\pi_{0,N}(W({\bf t}))$ of the moving point $W({\bf t})$ under the KP flows, and is thus given, up to projectivization, by the semi-infinite determinant~(\ref{tauWdef}).

More generally, $\hat{g}$ need not be a $GL(\HH)$ group element; it may be any element of the Clifford algebra satisfying the bilinear relation
\be
\left[\sum_{i\in \Zn} \psi_i \otimes \psi_i^\dag, \hat{g} \otimes \hat{g}\right] = 0
\label{bilinear}
\ee
acting upon $\FF \otimes \FF$. Eq.~(\ref{bilinear}) is equivalent to the Pl\"ucker relations and guarantees that $\hat{g}|0\rangle$ is a decomposable element,  as  in~(\ref{Pluckermap}).
In particular, (\ref{bilinear}) is satisfied  by any product of  pure creation or annihilation operators of the form
\be
w_a := \sum_{i \in \Zn} w_{ai} \psi_i, \qquad  v^\dag_a = \sum_{i \in \Zn} v_{ai} \psi^\dag_i.
\ee
More generally, we have the following useful result.

\begin{lemma}
 For any number of creation and annihilation operators $\{w_a, v_a^\dag\}_{a=1 \dots n}$, 
 if an element $\hat{g}$ satisfies the bilinear identity~(\ref{bilinear}), so does the product
\be
\left(\prod_{a=1}^n w_a\right) \left(\prod_{b=1}^n v^\dag_b\right) \hat{g}, 
\label{bilin_prod}
\ee
and hence
\be
\tau_{({\bf w}, {\bf v}, g)}:= \langle N \vert \hat{\gamma}_+({\bf t})\left(\prod_{a=1}^n w_a\right) \left(\ \prod_{b=1}^n v^\dag_b\right)  \hat{g} \vert N\rangle
\ee
is a KP $\tau$-function.
\end{lemma}

\begin{proof}
It follows immediately from the definitions that, if any two operators 
satisfy the bilinear identity, so does their product. Therefore, it is sufficient to prove it  holds for 
any creation operator $w_a$ or any annihilation operator $v_a^\dag$. Now, let $\mu, \nu \in \FF$ be 
a pair of elements  and apply the product  $\left(\sum_{i\in \Zn}\psi_i\otimes \psi_i^\dag\right) (w \otimes w)$
to  the decomposable element $\mu  \otimes \nu  \in \FF \otimes \FF$.
\bea
\sum_{i\in \Zn}(\psi_i\otimes \psi_i^\dag) (w_a \otimes w_a )\mu \otimes \nu 
&\& = \sum_{i\in \Zn}\psi_iw_a \mu \otimes \psi_i^\dag w_a \nu  \cr
&\& = \sum_{i\in \Zn} w_a \psi_i \mu \otimes w_a \psi_i^\dag  \nu   -  \sum_{i\in \Zn}\psi_i w_a \mu \otimes w_{ai} \nu \cr
&\& = \sum_{i\in \Zn} w_a \psi_i \mu \otimes w_a \psi_i^\dag  \nu   -  (w_a )^2\mu \otimes \nu \cr
&\& = \sum_{i\in \Zn} w_a \psi_i \mu \otimes w_a \psi_i^\dag  \nu,
\eea
where the anticommutation relations~(\ref{anticomm}) have been used, and the fact that $w_a^2=0$. Therefore
the bilinear relation~(\ref{bilinear})
 \be
\left[\sum_{i\in \Zn} \psi_i \otimes \psi_i^\dag,  w_a \otimes w_a \right] = 0
\ee
is satisfied by $w_a$. A similar calculation shows it holds for $v_b^\dag$, and hence for all products of the form~(\ref{bilin_prod}).
\end{proof}

%%%%%%%%%%%%%%%  Section 3 %%%%%%%%%%%%%%%%%%%

 \section{Reducing  infinite to finite Grassmannians }
 
 %%%%%%%%%%%%%%%  SubSection 3. 1 %%%%%%%%%%%%%%%%%%%

\subsection{ Grassmannian subquotients $W_1 \ss W  \ss W_2 \, \ra W_2/W_1 $}
 
 We now detail the subquotient reduction  described above. The  first step consists
 of choosing a pair of subspaces
 \be
 W_1 \ss W_2 \ss \HH, \quad W_1, \  W_2 \in \Gr_{\HH_+}(\HH)
 \ee
 of virtual dimension $(-n)$ and $k$ respectively, with $n+k =N$, both invariant under the action of the 
 abelian group  $\Gamma_+$ of KP flows
 \be
 \gamma_+({\bf t}) W_1 \ss W_1,  \quad  \gamma_+({\bf t}) W_2 \ss W_2, \quad \forall \gamma_+({\bf t}) \in \Gamma_+ 
 \ee
  so that
  \be
{\rm dim}\  W_2/W_1= n+k  = N .
 \ee
The element $W \in \Gr_{\HH_+}(\HH)$   is chosen to belong to the sub-Grassmannian
of virtual dimension $0$  subspaces that fit between them
 \be
 W_1 \ss W \ss W_2.
 \label{W1_W_W2}
\ee
Thus 
\be
 \dim\left(W/W_1\right) = n, \quad \dim\left(W_2/W\right) = k.
\ee

We now make an identification of the quotient $W_2/W_1$ with  $\Cn^N$ by choosing a subspace $W_1^{\cb} \ss W_2$
that is complementary to $W_1 \ss W_2$, choosing  a basis $\{b_1, b_2, \dots , b_N\}$
for it, and identifying this with the standard basis $\{f_1, f_2, \dots f_N\}$ for $\Cn^N$
\be
(f_i)_j =\delta_{ij}, \quad 1 \le i, j \le N.
\ee
Through the quotient map
\be
W_2 \ra W_2/W_1 \equiv \Cn^N,
\ee
any element $W\in \Gr_{\HH_+^0}(\HH)$ containing $W_1$ and contained in $W_2$ 
can be associated  with a unique maximal rank $n \times N$ matrix $C$ such that
\be
W = W_1 \oplus \span\left\{\sum_{i=1}^N C_{ai} b_i\right\}, \quad a= 1, \dots n.
\ee

The projection
\be
W \ra W/W_1  \ss W_2/W_1 \equiv \Cn^N
\ee
thus defines an element $[C]$ of the Grassmannian $\Gr_n(\Cn^N)$, spanned by the rows of $C$.
We denote this finite dimensional subquotient
\be
\Gr_n(W_2/W_1) \sim \Gr_n(\Cn^N)
\ee
   Since the  flow group $\Gamma_+$  leaves   both $W_1$ and $W_2$ invariant,
this induces an action of $\Gamma_+$ on $\Gr_n({W_2/W_1})$ such that,
through the identification $W_2/W_1 \sim\Cn^N$, the shift map $\Lambda: \HH \ra \HH$
 may be represented by an $N \times N$ matrix
\be
B: \Cn^N \ra \Cn^N
\ee
whose form depends on the choice of this basis, but whose Jordan canonical form depends only 
on the choice of pairs $(W_1,  W_2)$.
The  $\Gamma_+$ action induced on $\Gr_n(\Cn^N)$ will then be  represented by
\be
\gamma_+({\bf t}): C^T \mapsto e^{\sum_{i=1}^\infty t_i B^i}C^T := C^T({\bf t}).
\ee

To determine the associated $\tau$-function $\tau_W({\bf t})$, we must evaluate the 
 Pl\"ucker coordinate $\pi_0(W({\bf t}))$, which is the determinant of the projection operator
\be
W({\bf t}) \ra \HH_+.
\ee
 If $W_1$ is chosen to be a subspace of $\HH_+$, which in turn is contained  in $W_2$
\be 
W_1 \ss \HH_+ \ss W_2,
\ee
we may view the Grassmannian $\Gr_n(W_2/W_1)$ as an orbit
of the element $\HH_+/W_1 \in \Gr_n(W_2/W_1)$. The Pl\"ucker coordinate 
$\pi_0(C({\bf t}))$ relative to the given basis then coincides with $\pi_0(W({\bf t}))$, and
we may proceed in the same way as on the infinite Grassmannian $\Gr_{\HH_+^0}(\HH)$.
If, however, the inclusion condition
\be
\HH_+ \ss W_2
\ee
is not satisfied, a further transformation is needed to identify the reduced
Grassmannian $\Gr_n(W_2/W_1)$ as the orbit of some standard element
under $GL(W_2/W_1)$. This transformation determines the matrix $A$ in the
Gekhtman-Kasman formula~(\cite{GK1, GK2}).

  We now consider the case when $W_1, W_2$ are defined to be
  \be
  W_1 := r(z)\HH_+, \quad W_2 := {r(z)/ p(z) }\HH_+,
  \label{W1W2_rp}
  \ee
  where $r(z)$ and $p(z)$ are monic polynomials of degrees $n$ and $N$ respectively,
  with roots and multiplicities $\{\delta_i, n_i\}_{i=1, \dots, m}$, $\{ \beta_i, N_i\}_{i=1, \dots M}$
  \bea
  r(z) &\&:= \prod_{i=1}^m(z-\delta_i)^{n_i}, \quad p(z):= \prod_{j=1}^M(z-\beta_j)^{N_j},
  \cr
\sum_{i=1}^m n_i &\&= n,  \hskip 60 pt \sum_{j=1}^M N_j = N,
  \eea
 with the roots $\{\delta_i\}$ and  $\{\beta_j\}$ of  both $r(z)$ and $p(z)$ chosen to lie within the unit circle.
  (If they do not, we may just redefine the circle $S^1$ in  $\HH= L^2(S^1)$  as having a sufficiently large radius  that all roots of  are in the interior.) 
  $W_1$ is thus  the subspace of $\HH_+$ consisting of elements that vanish at the roots  
  $\{\delta_i\}_{i=1, \dots d}$ of $r(z)$ to the same order as their 
  multiplicities $\{n_i\}_{i=1, \dots d}$ in $r(z)$, while $W_2$ is the direct sum of $W_1$ 
  with the span of the rational basis elements
  \be
  b_{(j\nu)}(z) := {r(z)\over (z-\beta_j)^{\nu}}, \quad 1\le j \le M,  \  1\le \nu\le N_j.
  \label{bia_basis}
  \ee
  Note that if we use this basis to identify the quotient space $W_2/W_1$ with 
  $\Cn^N$, the matrix $B$ representing multiplication by $z$ is precisely the
  Jordan normal form matrix defined in  (\ref{B_jordan}).
    
  To complete the explicit matrix representation of the flows and Pl\"ucker coordinates,
  we must  identify the standard basis for $\Cn^n$ with a suitably chosen basis for
   $\HH_+/W_1$. We could of course choose this as the monomials of degree less
   than $n$, modulo $W_1$. But a more convenient choice   consists of 
 \be
d_{(i\mu)} (z) := { z^{\mu-1} \over (1-z\delta_i)^{\mu}}, \quad 1\le i \le m,  \  1\le \mu\le n_i .
 \label{ortho_basis}
  \ee 
These  are linearly independent elements of $\HH_+$ since the roots $\delta_i$ of $r(z)$ are distinct and lie within the unit circle.
Moreover the  $d_{(i\mu)} (z)$'s are orthogonal to  $W_1 = r(z) \HH_+$  with
respect to  the complex inner product  $( \ , \ )$ on $\HH$ in which the monomials are orthonormal,
which  may be expressed by the contour integral:
\be
(f, g) : = {1\over 2\pi i} \oint_{|z|=1} f(z^{-1}) g(z) {dz\over z},
\ee
since, by the Cauchy theorem,  
\be
(d_{(i\mu)}(z), r(z) z^a) = 0,  \quad \forall \  a\in \Nn.
\ee
They therefore form a basis for the orthogonal complement  $W_1^{\perp} \ss \HH_+$.
  The pairs $(i,\mu)$ may more concisely be labelled
  \be
  I := (i, \mu), \quad 1\le i \le m, \ 1 \le \mu \le n_i
  \label{I_i_mu}
  \ee
  ordered consecutively as
   \be
   \label{eq:multiind_order}
  (1,1), \dots (1,n_1), \dots (i, 1), \dots (i, n_i), \dots (m, 1), \dots (m, n_m).
  \ee

 We  assume henceforth that the roots $\{\delta_i\}_{i=1, \dots m}$ and $\{\beta_j\}_{j=1, \dots , M}$ of the polynomials $r(z)$ and $p(z)$  are {\it all} distinct  and interpret these  as  characteristic polynomials of the pair of matrices 
$B \in \Mat{N }{N}$,  
$D\in \Mat{n }{n}$ defined in  (\ref{B_jordan}) and (\ref{D_jordan}), 
\bea
\label{r_Bz}
p(z) &\&= r_B(z) := \det(z\,\Id{N} - B) = \prod_{j=1}^M(z-\beta_j)^{N_j} ,\\
 r(z) &\& = r_D(z):= \det(z\,\Id{n} - D)=\prod_{i=1}^m(z-\delta_i)^{n_i}. 
 \label{r_Dz}
\eea

  Denoting the pairs of  indices $\{(j,\nu)\}$ by
  \be 
  J := (j,\nu), \quad 1\le j \le M,  \  1\le \nu \le N_j,
  \ee
ordered again consecutively as
   \be
   \label{eq:multiind_order}
  (1,1), \dots (1,N_1), \dots (j, 1), \dots (j ,N_j), \dots (M, 1), \dots (M, N_M).
  \ee 
  The basis elements  for $W_2 / W_1$, labelled accordingly, will be denoted $ b_J$, and the elements of the $n\times N$ matrix $C \in \Mat{n }{N}$ as $C_{a,J}$ 
 with $a =1, \dots ,n$ or, when a refinement is needed, as $C_{I, J}$, with $I$ defined as  in
 (\ref{I_i_mu}). The sub-Grassmannian $\Gr_n(W_2/W_1)$ then consists of all $W$'s of the form
  \be
  W(B, C, D) := W_1 \oplus \span \left(\sum_{J} C_{I, J} b_J\right)_{1\le a \le n}.
  \label{WbI}
  \ee

  \br {\bf Gauge equivalence.} Gauge transformations~(\ref{gaugetransform}) that preserve the 
class of subspaces $W_1$, $W_2$ of type~(\ref{W1W2_rp})  consist of multiplication
  by  rational functions that take value $1$ at $z = \infty$; i.e, the ratio of two monic polynomials
  $q(z)$, $\tilde{q}(z)$ of the same degree
  \be
\gamma_-( {\bf s} ) = {q(z) \over \tilde{q}(z)}, \quad \deg(q) = \deg(\tilde{q}).
\ee
  We can therefore use gauge transformations to replace $r(z)$ by any rational function whose singular part 
  at $z=\infty$ is a polynomial of degree $n$.
    \er
  \br
  The union of all sub-Grassmannians $\Gr_n(W_2/W_1)$ over all choices of $(n, N)$, $N>n$,
  and all polynomials $r(z)$, $p(z)$  is essentially the virtual dimension $0$ component of the Grassmannian $\Gr_1$ defined in~\cite{SW}. More precisely, $\Gr_1$ consists of those $W$'s corresponding to the choices
 \be
  W_1 := r(z)\HH_+, \quad W_2 := {1\over q(z)}\HH_+,
  \label{W1W2_SW}
  \ee
for a polynomial $q(z)$ of degree $k$. But these are easily seen to be gauge equivalent to the choice~(\ref{W1W2_rp}). In fact, within gauge transformations, there is no loss of generality in choosing $r(z)$ as the monomial $z^n$; i.e., choosing the subspaces $W_1$, $W_2$ as
  \be
  W_1 = z^n \HH_+, \quad W_2 = {z^n \over p(z)}\HH_+.
  \label{znW1}
  \ee

  Note also that, from the viewpoint of the KP hierarchy, nothing new is
  added by considering $W$'s with virtual dimension different from $0$, since
  there is always an equivalent $\tau$-function in the zero charge sector. The charge sector only
  becomes relevant when considering lattices of $\tau$-functions having
   the same group element $\hat{g}$ in~(\ref{VEV_tau}), with the lattice site 
 given by the fermionic charge (or virtual dimension), thereby defining elements
  of  an infinite flag manifold~\cite{DJKM1, DJKM2}.
   \er
  
   %%%%%%%%%%%%%%%  SubSection 3. 2 %%%%%%%%%%%%%%%%%%%

\subsection{KP $\tau$-functions as finite determinants}

Define the $n \times N$ matrix $A(B,D)\in \Mat{n }{N} $ by the following  formula:
\bea
\label{ABD}
A_{(i\mu), (j\nu)}(B,D)  &\&:= {1\over 2\pi i} \oint_{|z|=1}{ r_D(z) \over (z-\beta_j)^{\nu} (z-\delta_i)^{\mu}} dz,\\
&\& \cr
\quad 1 \le i \le m, &\&  \ 1\le \mu \le n_i, \ 1\le j \le M, \ 1\le \nu \le N_j. 
\nonumber
\eea
and the further $n \times N$ matrix  $A^0(B, D)\in \Mat{n }{N} $ by the following:
\be
A^0_{(i\mu), (j\nu)}(B, D) := { \mu+\nu -2 \choose \nu-1 } { (-1)^ {\nu+1}\over (\beta_j - \delta _i)^{\mu + \nu -1}}. 
\label{A0BD}
\ee
Evaluating the integral using the Cauchy formula, it follows that these are related by 
right multiplication by the matrix $r_D(B)$:
\begin{lemma}
\be
A(B,D) = A^0(B,D) r_{D}(B).
\label{A0RBD}
\ee
\end{lemma}
\begin{proof} 
From the Cauchy residue formula and  Leibniz' rule we
have
\bea
A_{(i\mu), (j\nu)}(B,D) 
  &\& = \left. {1\over ( \nu-1)!} {d^{\nu-1}\over d z^{\nu-1} } \right|_{z=\beta_j}\left( r_D(z) {1\over (z-\delta_i)^{\mu}}\right)   \cr
 &\& =\sum_{\xi=1}^{\nu} {r_D^{(\nu-\xi)}(\beta_j)\over (\nu-\xi)!} {\mu+\xi-2 \choose \xi-1} {(-1)^{\xi-1} \over (\beta_j - \delta_i)^{\mu+\xi-1}} \cr
&\&=\sum_{k=1}^M \sum_{\xi=1}^\nu A^0_{(i\mu), (k\xi)} (B,D) (r_{D}(B))_{(k\xi),(j\nu)},
  \eea
 since 
\be
(r_{D}(B))_{(k\xi),(j\nu)}= \delta_{jk} {r_D^{(\nu-\xi)}(\beta_j) \over (\nu-\xi)!} \quad {\rm for} \ \nu \ge \xi.
\ee
\end{proof}

Denote by 
\be
A(B):= A(B, \Lambda_n), 
\ee
the particular case where the matrix $D$ is chosen as $\Lambda_n$.
The matrix elements of $A(B)$ are easily computed to be
\bea
A_{a, (j\nu)}(B) &\&= {1\over 2\pi i}\oint_{|z|=1} {z^{n-a} \over (z-\beta_j)^{\nu}} dz = { n-a \choose \nu-1} \beta_j^{n-\nu-a+1},
\label{A_B_me}
\\
 1\le a \le n, &\& \quad 1\le j \le M, \quad 1\le \nu \le N_j 
 \nonumber
\eea
We also define the $n \times n$ matrix $K(D) \in \Mat{n }{n}$ with elements
\bea
K_{a, (i\mu)}(D)&\&:= {1\over 2\pi i} \oint_{z= \delta_i}  {z^{n-a} (z- \delta_i)^{\mu-1}\over r_D(z)}dz, 
\\
 1\le a \le n, &\& \quad 1\le i \le m, \quad 1\le \mu \le n_i ,
\eea
where the  integral is over any counterclockwise contour containing only the pole $z=\delta_i$.
The following gives the relation between these matrices
\begin{lemma}
\be
A(B) = K(D) A(B,D).
\label{KABD}  
\ee
\end{lemma}
\begin{proof} Use  Lagrange interpolation  to express
\be
z^{n-c} = \sum_{i,\mu} K_{c, (i\mu)} (D){ r_D(z) \over  (z -\delta_i)^{\mu} }
\ee
and substitute into (\ref{A_B_me}).
\end{proof}

Let $\eg{B} \in \Cn^N$, $\ef{D} \in \Cn^n$ and $\bf k (D)\in \Cn^n$ be the vectors:
 \bea
 \eg{B}_{(j\nu)} &\&=\delta_{\nu,1},   \quad \quad 1\le j \le M, \quad 1\le \nu \le N_j,  \\
 \ef{D}_{(i\mu)} &\&=\delta_{\mu,1},   \quad \quad 1\le i \le m, \quad 1\le \mu \le n_j,  \\
 {\bf k}_{(i\mu)}(D)  &\&=  {1\over 2\pi i} \oint _{z=\delta_i} {z^n (z-\delta_i)^{\mu - 1}\over r_D(z)} dz.
  \eea
    We then have the following  identities:
  \begin{lemma}
  \bea
    \label{Lambda_K}
  \Lambda^T_n K(D) &\&= K(D) D^T - \ef{\Lambda_n} {\bf k}^T (D)\\
  K(D) \ef{D} &\&= \ef{\Lambda_n}
  \label{KeD}
  \eea
  \end{lemma}
\begin{proof} These follow from the definitions and the contour integral
\be
\oint_{|z|=1} {z^{n-c} \over r_D(z)} dz = \delta_{c1},  \quad 1\le c \le n.
\ee
\end{proof}

The following shows that the matrices $A^0(B,D)$ and $A(B,D)$ both satisfy the rank-$1$ condition
(\ref{rank1}) for the same pair of matrices $B$ and $D$, but with different RHS.
\begin{proposition}
\label{prop_AB_rank1}
\bea
\label{AB_rank1}
A(B) B - \Lambda_n^T A(B) &\&= \ef{\Lambda_n} \eg{B}^{T} B^n  \\
\label{A0BD_rank1}
A^0(B,D) B - D^T A^0(B,D) &\& = \ef{D} \eg{B}^{T}\\
\label{ABD_rank1}
A(B,D) B - D^T A(B,D) &\& = \ef{D} \eg{B}^{T}r_{D}(B)
\eea
\end{proposition}
\begin{proof} 
Eq.~(\ref{AB_rank1}) follows from the definitions (and is also the
special case of (\ref{ABD_rank1}) when $D=\Lambda_n$).
  Eqs.~(\ref{A0BD_rank1}) and (\ref{ABD_rank1}) follow from 
  substituting (\ref{A0RBD}), (\ref{KABD}), (\ref{Lambda_K}) and (\ref{KeD}) into (\ref{AB_rank1}) and using
  the fact that $B$ commutes with $r_D(B)$. 
  
   Eq.~(\ref{A0BD_rank1}) can also be proved directly  as follows. 
   Since both $B$ and $D$ are in standard Jordan normal form, 
   we may subdivide $A^0(B,D)$ into
  $m \times M$ blocks of  sizes $n_i \times N_j$
  \be
  A^0(B,D) = \left( \AA_{ij}\right), \quad 1\le i \le m, \quad 1\le j \le M, 
\ee
  where each $n_i \times N_j$ block  is of the form
  \be
  \AA_{\mu\nu} (\beta, \delta) = \pmatrix{ \mu+\nu -2 \cr \nu-1} {(-1)^{\nu} \over (\beta -\delta)^{\mu+\nu -1}}
  \ee
  with
  \be
   \delta = \delta_i, \quad \beta = \beta_j,  \quad 1\le i \le m , \quad 1\le j \le M, 
  \quad 1\le \mu \le n_i, \quad 1\le \nu \le N_j. 
  \ee
  Denoting  the elementary Jordan blocks with  matrix elements
  \be
  J_{\mu\nu}(\beta)= \beta \delta_{\mu\nu} + \delta_{\mu+1, \nu}, \quad J_{\xi,\kappa}(\delta) = \delta \delta_{\xi,\kappa} + \delta_{\xi+1, \kappa} ,
  \ee
as $J(\beta)$ and $J(\delta)$ with the appropriate range of indices $\mu,\nu,\xi,\kappa$ given by the dimension of the Jordan blocks,  it is easily verified that
  \be
   \left( \AA (\beta, \delta) J_B - J^T_D(\delta)\AA(\beta, \delta)\right)_{\mu \nu} = \delta_{\mu 1} \delta_{\nu 1}
  \ee
  Applying this to each $n_i \times N_j$ block is equivalent to   eq.~(\ref{A0BD_rank1}).
  
 \end{proof}
\br In terms of the resolvents $(z\Id{n}-D^{T})^{-1}$ and $(z\Id{N}-B)^{-1}$, (\ref{ABD}) can be written equivalently as
\be
\label{eq:ABDbiresolvent}
A(B,D) = \frac{1}{2\pi i}\oint_{|z|=1}r_D(z)(z\Id{n}-D^{T})^{-1}\ef{D}\eg{B}^{T}(z\Id{N}-B)^{-1}dz.
\ee
Also, note that
\be
K(D)r_D(z)(z\Id{n}-D^{T})^{-1}\ef{D}=r_{\Lambda_n}(z)(z\Id{n}-\Lambda_n^{T})^{-1}\ef{\Lambda_n},
\ee
which explains (\ref{KABD}) through (\ref{eq:ABDbiresolvent}).
\er

The  $\tau$-function $\tau_{W(B, C,D)}({\bf t})$ corresponding to the element 
$W(B, C, D) \in \Gr_{\HH_+^0}(\HH)$  defined in~(\ref{WbI}), with $W_1$, $W_2$ defined in~(\ref{W1W2_rp}), 
is given by the following:

\begin{theorem}
\label{theorem_tau_BCD}
\bea
\tau_{W(B, C, D)}({\bf t}) &\&=  \det(A(B, D) e^{\sum_{I=1}^\infty t_i B^i} C^T)
\label{tau_BCD}
 \\
&\& = \det(A^0(B, D) e^{\sum_{I=1}^\infty t_i B^i} r_{D}(B)  C^T)
\label{tau_A0BDR}
\\
&\& =\kappa(D)^{-1}   \det(A(B) e^{\sum_{I=1}^\infty t_i B^i}  C^T)
\label{tau_AB},
\eea
where
\be
\kappa(D) := \det(K(D)).
\ee
\end{theorem}
\begin{proof} 
 We begin with the case $D= \Lambda_n$, $r_D(z)=z^n$ and  $B $ of
 the diagonal form~(\ref{diagB}); that is, all Jordan blocks of $B$ have dimension $N_j =1$. 
  For this case, a basis for the space $W(B,C, \Lambda_n)$ may be taken as
\be
\{q_a,  z^{n+i}\}_{a=1, \dots , n}, \quad 
 i\in \Nn, 
\ee
where
\be
q_a :=  \sum_{j=1}^N C_{a, j}{z^n\over z-\beta_j}
\ee
Since the subspace $z^n\HH_+ \ss W$ is invariant under the $\Gamma_+$ action,
we may choose 
\be
\{{q_a({\bf t}), z^{n+i} \}_{a=1, \dots , n},  \quad  i\in \Nn}
\ee
as basis  for the orbit space
\be
W({\bf t})  := \gamma_+({\bf t}) W,
\ee 
where
\bea
q_a({\bf t}) &\&:=  \sum_{j=1}^N {C_{a, j} e^{\sum_{i=1}^\infty t_i \beta_j ^i} z^n\over z-\beta_j} \cr
&\&\phantom{:} = \sum_{j=1}^N {C_{a, j} ({\bf t})z^n\over z-\beta_j}.
\eea
Here $\{C_{a, j} ({\bf t})\}$ are the matrix elements of
\be
\label{C_t}
C({\bf t}) := C e^{\sum_{i=1}^\infty t_i B^i}
\ee
The $\tau$-function $\tau_{W(B, C, \Lambda_n)}({\bf t})$ is the determinant of the orthogonal projection to $\HH_+$; that is, of the  infinite matrix  $\MM$ of scalar products of the basis elements  $\{z^i\}_{i\in \Nn}$ of $\HH_+$ with those of $W(B, C, \Lambda_n)$.
We have, for $1\le i, a \le n$,
\be
\MM_{i,a} = \left (z^{i-1}, q_a({\bf t})\right)= 
 {1\over 2\pi i}\oint_{|z|=1} {dz}  \sum_{j=1}^N C_{a, j}({\bf t}) {z^{n-i}\over z -\beta_j} =   \sum_{j=1}^N C_{a, j}({\bf t})\beta_j^{n-i}
\ee
and
\be
\MM_{n+ i,j} = \MM_{j,n+i}  = \delta_{n+i, j}, \quad i, j \in \Nn^+.
 \ee
Since $\MM$ is block triangular,  only its first $N \times N$ block contributes to its determinant, and 
 therefore
 \be
 \tau_{W(B, C, \Lambda_n)}({\bf t}) = \det (\MM )= \det \left(A(B) e^{\sum_{i=1}^\infty t_iB^i} C^T\right),
 \ee
 where $A(B) = V_{n, N} (\beta)$ is the truncated Vandermonde matrix defined in~(\ref{VdM_nN}).

The case where $B$ has general Jordan blocks of dimensions $\{N_j\}_{j=1, \dots M}$ and $D= \Lambda_n$ is proved in
exactly the same way. The basis elements $q_a ({\bf t})$ are replaced by
\be
q_a({\bf t}) := \sum_{j=1}^M\sum_{\nu=1}^{N_j}
C_{a, (j, \nu)}({\bf t}){ z^n \over (z-\beta_j)^\nu}, 
\ee
where  $C_{a, (j, \nu)}({\bf t})$ are the matrix elements of
\be
C({\bf t}) =C  e^{\sum_{i=1}^\infty t_i (B^T)^i }
\ee
and $B$ is the $N\times N$ matrix of  general nondegenerate Jordan form defined in
eq.~(\ref{B_jordan}). The residue calculation is carried out in the same way, with the higher order poles 
giving derivatives of the columns of the truncated Vandermonde matrix~(\ref{VdM_nN})
 to the same order as the pole. The resulting form for the matrix $A$ is the generalized
 truncated Vandermonde matrix  $V'_{n, N}$ with elements
 \be
A_{a, (j, \nu)} ={(n - a +1 )! \over (n- a - \nu+1)!}  (\beta_j)^{n - a - \nu +1 },
\ee
which is the matrix $A(B,D)$ in eq~(\ref{ABD}) for the special case  $m=1, \ \delta_1=0$.

To prove the general case, where both $B$ and $D$ have general Jordan block structure  (\ref{B_jordan})
and (\ref{D_jordan}), we proceed in the same way, but replace
the monomials $\{z^i\}_{i=0, \dots n-1}$ spanning the orthogonal complement to $z^n\HH_+$
in $\HH_+$ by the elements $\{d_{i\mu}(z)\}$  defined in (\ref{ortho_basis})
that span the orthogonal complement to $W_1=r(z) \HH_+$ and the basis elements spanning $W/W_1$ 
and $W({\bf t})/W_1$  by
\be
 q_{i \mu} := \sum_{j=1}^M\sum_{\nu =1}^{N_j}  C_{i \mu, j \nu} { r_D(z)  \over (z - \beta_j)^\nu} ,
  \quad 1\le i \le m , \  1\le \mu \le n_i .
\ee
and
\be
 q_{i \mu}({\bf t}) := \sum_{j=1}^M\sum_{\nu =1}^{N_j}  C_{i \mu, j \nu}({\bf t}){ r_D(z)  \over (z - \beta_j)^\nu} ,
  \quad 1\le i \le m , \  1\le \mu \le n_i .
\ee
Computation of the determinant of the orthogonal projection  leads to the evaluation
of the contour integrals
\be
{1\over 2\pi i} \oint dz {r_D(z) \over (z-\delta_i)^\mu (z-\beta_j)^\nu}   = A_{_{(i\mu), (j\nu)} }(B,D),
\ee
resulting in the general form (\ref{tau_BCD}) and (\ref{tau_AB}).  Equalities (\ref{tau_A0BDR}) 
and (\ref{tau_AB}) then follow from (\ref{A0RBD}) and (\ref{KABD}). 
\end{proof}
\br
The expressions (\ref{tau_BCD})-(\ref{tau_AB}) can be rewritten equivalently as
\be
\tau_{W(B,C,D)} = \det(P(D))^{-1}\det(\tilde{A}e^{\sum_{i=1}^\infty t_i B^i}\tilde{C}^T)
\ee
where
\be
\tilde{C}^T =(R(B)^T)^{-1} C^T, \quad \tilde{A}= P(D)^T A(B,D) R(B)
\label{Ktilde_Rtilde_def}
\ee
and $P(D) \in GL(n)$ and $R(B) \in GL(N)$ are arbitrary elements of the stabilizers of 
$D$ and $B$, respectively, under conjugation.
The rank-1 condition satisfied by the resulting matrix $\tilde{A}$ is then
\be
\tilde{A}B - D^T \tilde{A} = {\vf} {\vg}^T
\label{Atilde_rank1}
\ee
where
\be
{\vf} = P(D)^T \ef{D}, \quad {\vg} = R(B) \eg{B}
\label{fgKRdef}
\ee
are now arbitrary vectors determined by a suitable choice of $P(D)$, $R(B)$, satisfying
 the generic conditions
\be
\label{eq:generic_vec}
f_{i, 1} \ne 0, \ g_{j, 1} \ne 0,  \quad 1\le  i=1 \ \le m, \ 1 \le j \le  M.
\ee
Those cases with degeneracies violating (\ref{eq:generic_vec}) will be dealt with in a subsequent publication.
\er

    The solutions of this type are ``generic'', in the sense that they form an open dense set defined
by the requirements that that eigenvalues $(\delta_1, \dots, \delta_m)$ and $(\beta_1, \dots , \beta_M)$ 
be distinct and the conditions (\ref{eq:generic_vec}) be satisfied. In the sequel \cite{BFH}. a complete classification
of all solutions of the rank-1 condition (\ref{rank1_D}), will given, for $B$ and $D$ having arbitrary Jordan forms,
within equivalence under the action of their stabilizers.
\br
 Eq.~(\ref{KABD}) shows that the choice of  the matrix $D$ is in fact immaterial,
since it only affects the $\tau$-function by a constant multiplicative factor. This is a little surprising,
since a change in $D$ is equivalent to a change in $r(z)$, and hence is actually a gauge transformation 
that should give rise to a multiplicative linear exponential factor. The absence of this factor is
due to the fact that an ``admissible section'' must be chosen in the dual determinantal line 
bundle (cf. \cite{SW}) when considering the lift of the $\Gamma_{\pm}$ actions from the Grassmannian. 
 From this it follows that, whereas the two abelian subgroups $\Gamma_+$ and $\Gamma_-$, 
when  acting upon the Hilbert space $\HH$ mutually commute, the actions  induced on the holomorphic
sections of the  dual determinantal line bundle over the Grassmannian, only commute within a scalar multiplicative factor;
i.e., there is a central extension. 

 To obtain this in the finite Grassmannian approach, we recall that the notion of the determinant of the 
 projection map $W({\bf t}) \ra \HH_+$ is only well defined if a basis is chosen 
for both $W({\bf t})$ and $\HH_+$ allowing us to view this as an endomorphism. In the proof of (\ref{tau_BCD}),
we have chosen the basis $\{ r_D(z)\HH_+ z^j, q_{i  \mu}\}_{j\in \Nn, \ 1\le i \le m , \  1\le \mu \le n_i}$ for
$W(B,C,D)$ and $\{ r_D(z)\HH_+ z^j, d_{i  \mu}\}_{j\in \Nn, \ 1\le i \le m , \  1\le \mu \le n_i}$ for $
\HH_+$. 
To obtain the missing gauge factor $e^{-\sum_{i=1}^\infty t_i \tr(D)^i}$, as in eq. (\ref{tau_BCD_gauge}), Theorem \ref{theorem_fermionic_tau_BCD}  below,  we may choose instead the time dependent basis  $\{ r_D(z)\HH_+ z^j, \gamma_+({\bf t})d_{i  \mu}(z)\}_{j\in \Nn, \ 1\le i \le m , \  1\le \mu \le n_i}$
for $\HH_+$.
The corresponding  calculation is made in the next subsection using the  fermionic representation, in which this gauge factor appears automatically, since the fermionic  Fock space is precisely the space of   admissible holomorphic  sections of the dual determinantal line bundle.

\er

%%%%%%%%%%%%%%%  subsection  3.3  %%%%%%%%%%%%%%%%%%%

\subsection{Fermionic representation}

We now give  an alternative derivation of  formula~(\ref{tau_BCD}) using fermionic operators. 
Define the following Fermi creation operators
 \bea
 \label{w_I_C}
 w_{(i, \mu)} (B, C, D) &\&:= \sum_{j=1}^M\sum_{\nu=1}^{N_j} C_{(i, \mu),(j,  \nu)} \Psi_D^{\nu} (\beta_j),\\
i  =1, \dots , m, &\& \quad \mu = 1, \dots , n_i,
 \nonumber
 \eea
 where 
  \be
 \Psi_D^{\nu} (\beta_j) := {1\over 2 \pi i } \oint_{z= \beta_j}{ r_D(z) \psi(z)\over\left(z-\beta_j\right)^\nu} dz
 ={1\over (\nu-1)!}{\partial^\nu \over \partial \beta_j^\nu}(\psi(\beta_j) r_{D}(\beta_j)), 
 \ee
 with the integral taken over a small contour containing only the pole at $z=\beta_i$. 
  Similarly, define the annihilation operators
   \be
 \left(\Psi^{\mu } (\delta_i)\right)^\dag := {1\over 2 \pi i } \oint_{z= \delta_i}{ \psi^\dag(z) \over\left(z-\delta_j\right)^\mu}  dz
 = {1\over (\mu-1)!} {\partial^{\mu-1} \over \partial \delta_j^{\mu-1}}\psi(\delta_j) .
 \ee
\begin{lemma}
\label{plucker_W_W1_W2}
The images of the subspaces $W_1$, $W_2$ and $W$ under the Pl\"ucker map~(\ref{Pluckermap}),  are
\bea
\label{Pl_W1}
\grP\grl(W_1) &\&=   \prod_{i=1}^m \prod_{\mu=1}^{n_i } \left(\Psi^{\mu}(\delta_i)\right)^{\dagger} |0 \rangle, 
\\
\label{Pl_W2}
\grP\grl(W_2) &\&=    \prod_{j=1}^M \prod_{\nu=1}^{N_j } \Psi_D^{\nu} (\beta_j) 
 \prod_{i=1}^m \prod_{\mu=1}^{n_i } \left(\Psi^{\mu}(\delta_i)\right)^{\dagger}|0 \rangle,
\\
\label{Pl_W}
\grP\grl(W) &\&= \prod_{i=1}^m\prod_{\nu=1}^{n_j} w_{(j, \nu)} (B,C, D)   
 \prod_{i=1}^m \prod_{\mu=1}^{n_i } \left(\Psi^{\mu}(\delta_i)\right)^{\dagger} |0 \rangle .
\eea
\end{lemma}

\begin{proof} 
We begin with~(\ref{Pl_W1}). Viewing $\{\psi_j^\dag \sim e^j\}$ as elements of the dual space $\HH^*$, acting on $\FF$ through the Clifford representation by inner products, 
as in eq.~(\ref{psi_i_psidag_i}), a basis for the annihilator 
of $W_1 \ss \HH_+$ consists of   $\{(\Psi^{\mu}(\delta_i))^\dag\}_{i=1, \dots d;\ \mu=1, \dots n_i}$. The latter follows from the fact that $(\Psi^{\mu )}(\delta))^\dag$,  as a linear form on $\HH$,
acting upon an element $f(z) \in \HH_+$, under the identification $z^i \sim e_{-i-1}$, evaluates to its derivative at $z=\delta$:
\be
(\Psi^\mu (\delta))^\dag( f(z)) = {1\over (\mu -1)!}f^{(\mu -1)}(\delta).
\ee
Eq.~(\ref{Pl_W1}) follows, since the image $\grP\grl(U)$, of any subspace $U \ss \HH_+$
under the Pl\"ucker map $\grP\grl: U \ra \FF$  is the joint kernel of the
elements of its annihilator within the Clifford representation which, for $U=W_1$, is given by the r.h.s.\ of eq.~(\ref{Pl_W1}).

To prove~(\ref{Pl_W2}), we note that the subspace $W_2$ is obtained by extending $W_1$ by the basis elements $\{q_{j \nu}\}_{1\le j\le M , \ 1 \le \nu \le N_j}$. The Pl\"ucker image of $W_2$ is therefore obtained by applying the wedge product of the elements $\{q_{j \nu}\}$ to the Pl\"ucker image (\ref{Pl_W1}) of $W_1$. But this is equivalent to applying
the product of the operators $\Psi_D^{\nu} (\beta_j)$, since each of these may be expressed as the exterior product with 
$q_{j\, \nu-1}$ plus a linear combination of lower order terms $\{q_{j, \nu -i}\}$, $i=2, \dots, \nu$.
Its Pl\"ucker image is therefore given by~(\ref{Pl_W2}).

Finally, to obtain (\ref{Pl_W}) as the Pl\"ucker image of $W$, we replace the product of creation operators 
$\prod_{\nu=1}^{N_j } \Psi_D^{\nu} (\beta_j) $ in (\ref{Pl_W2}) by the product
$ \prod_{i=1}^m \prod_{\mu=1}^{n_i } \left(\Psi^{\mu}(\delta_i)\right)^{\dagger}$ corresponding
to the basis elements  $\{q_{i \mu}\}_{1 \le i \le m, \  1 \le \mu \le n_i}$ that complete the basis for $W$.
\end{proof}

\medskip
From this lemma and the equivariance of the Pl\"ucker map,  it follows that the $\tau$-function of  
Theorem~\ref{theorem_tau_BCD},   
eq.~(\ref{tau_BCD}), may equivalently be expressed in fermionic form as:
\begin{theorem}
\label{theorem_fermionic_tau_BCD}
 \bea
  \langle 0 | \hat{\gamma}_+ ({\bf t}) \prod_{j=1}^m\prod_{\nu=1}^{N_j} w_{(j, \nu)} (B,C)  
  \prod_{i=1}^m \prod_{ \mu =1}^{n_i}   \left(\Psi^{\mu }(\delta_i)\right)^{\dagger}  | 0 \rangle  
  &\&=   e^{-\sum_{i=1}^\infty t_i \tr(D)^i} \tau_{W(B, C, D)} ({\bf t}) 
   \label{fermionic_tau_BCD} \cr
=  e^{-\sum_{i=1}^\infty t_i \tr(D)^i}\det\left(A(B, D) e^{\sum_{i=1}^\infty t_i B^i} C^T\right)  &\&.
  \label{tau_BCD_gauge}
\eea
\end{theorem}
\begin{proof}  
The proof is based on  Wick's identity, starting from the fermionic expression~(\ref{fermionic_tau_BCD}).
It follows from Lemma~\ref{plucker_W_W1_W2} that $\tau_{W(B, C, D)}({\bf t}) $ is given by
\be
\label{tau_WBCD_fermi}
\tau_{W(B,C,D)}({\bf t})  = \langle 0 |\hat{ \gamma}_+({\bf t})
\grP\grl(W)  \rangle =\langle 0 | \hat{ \gamma}_+({\bf t}) \prod_{i=1}^m \prod_{\mu=1}^{n_i }  
\left(\Psi^{\mu}(\delta_i)\right)^{\dagger} \prod_{j=1}^m\prod_{\nu=1}^{n_j} w_{(j, \nu)} (B,C)   |0 \rangle .
\ee
We first note that,  since the group $\Gamma_+$ stabilizes the vacuum,  left multiplication by the element  
$\hat{\gamma}_+({\bf t})$  in (\ref{tau_WBCD_fermi}) is equivalent to conjugation of all the Fermi creation and annihilation operators:
\bea
(\Psi^{\mu}(\delta_i))^\dag  \ra  \hat{\gamma}_+({\bf t}) (\Psi^{\mu}(\delta_i))^\dag
 \hat{\gamma}_+^{-1}({\bf t}) 
 &\&= { 1\over (\mu-1)!} {\partial^{\mu-1}\over \partial \delta_i ^{\mu -1}}
 \left( e^{-\sum_{k=1}^\infty t_k \delta_i^k} \ \psi(\delta_i)^\dag \right), 
  \cr
 \Psi^\nu_D(\beta_j)  \ra   \hat{\gamma}_+({\bf t}) \Psi^\nu(\beta_i)
 \hat{\gamma}_+^{-1}({\bf t}) 
 &\&= { 1\over (\nu-1)!}{\partial^{\nu-1}\over \partial \beta_j ^{\nu -1}}
 \left( e^{\sum_{k=1}^\infty t_k \beta_j^k} \ \psi(\beta_j)r_D(\beta_j) \right). 
 \eea
The net effect is to multiply the vacuum expectation value of the terms
without the $\hat{\gamma}_+({\bf t}) $ factor by an overall linear exponential factor
$e^{-\sum_{k=1}^\infty {t_k}\tr(D^k)}$ and replace the matrix $C$ in the expression~(\ref{w_I_C}) for $w_{(i, \mu)}(B, C, D)$ by $C({\bf t})$, as in eq.~(\ref{C_t}).
\be
\tau_{W(B,C, D)}({\bf t})  =e^{-\sum_{k=1}^\infty {t_k}\tr(D^k)}
\langle 0 |  \prod_{i=1}^m \prod_{\mu=1}^{n_i }  
\left(\Psi^\mu(\delta_i)\right)^{\dagger}  \prod_{j=1}^m\prod_{\nu=1}^{n_j} w_{(j, \nu)} \left(B,C({\bf t}) \right)  |0 \rangle .
\ee
By Wick's identity, the vacuum matrix element can be written as an $n \times n$ determinant
\bea
\langle 0 |  \prod_{i=1}^m \prod_{\mu=1}^{n_i }  
\left(\Psi^\mu(\delta_i)\right)^{\dagger}  \prod_{j=1}^m\prod_{\nu=1}^{n_j} w_{(j, \nu)} \left(B,C({\bf t}) \right)  |0 \rangle 
&\& = \det \langle 0 |  \left(\Psi^\mu(\delta_i)\right)^{\dagger}  w_{(j, \nu)} \left(B,C({\bf t}) \right)  |0 \rangle \cr
 &\&= \det(A(B,D) C^T({\bf t})),
\eea
since
\bea
\langle 0 |  \Psi^\nu(\beta_j)(\Psi^\mu(\delta_i))^\dag |0\rangle 
&\&= {1\over 2\pi i} { 1\over (\mu-1)!} {\partial^{\mu-1}\over \partial \delta_i^{\mu -1}}
\oint_{z=\beta_j}{\langle 0 | \psi(z) \psi^\dag(\delta_i) | 0 \rangle r_D(z)\over (z-\beta_j)^\nu} dz \cr
&\& = {1\over 2\pi i} { 1\over (\mu-1)!} {\partial^{\mu-1}\over \partial \delta_i^{\mu -1}}
\oint_{z=\beta_j}{ r_D(z)\over (z-\beta_j)^\nu (z-\delta_i)} dz \cr
&\& ={1\over 2\pi i}
\oint_{z=\beta_j}{ r_D(z)\over (z-\beta_j)^\nu (z-\delta_i)^\mu} dz \cr
&\& = A_{(i\mu), (j\nu)}(B, D).
\eea
\end{proof}

    A slightly more general class of $\tau$-functions having a finite dimensional exponential determinantal form may be constructed as follows. For three positive integers $l, n, N$ with $l\le n$, $l \le N$, but no restriction relating $n$ and $N$, we again choose a pair  $D\in \Mat{n}{n}$, $B \in \Mat{N}{N}$ of  square matrices, and a rectangular matrix $A\in \Mat{n }{N}$ satisfying the rank-1 condition (\ref{rank1_D})  for a pair of vectors  ${\vf} \in  \Cn^n$,  ${\vg}\in \Cn^N$. 
Then choosing any pair $F\in \Mat{l }{n}$,  $C\in \Mat{l}{N}$ such that $FAC^T$ is invertible, the following $l \times l$ determinantal formula
\be
\tau^f_{(A,B,C,D,F)}({\bf t}):= \det(F e^{-\sum_{i=1}^\infty t_i (D^T)^i} A  e^{\sum_{j=1}^\infty t_jB^j}C^T)
\label{tau_ABCDF}
\ee
is shown to define a KP $\tau$-function in the next section. The appendix gives a simple direct verification that $\tau^f_{A,B,C,D,F}({\bf t})$ satisfies the Hirota bilinear relations.
 
  Assuming the eigenvalues of $B$ and $D$ to be mutually distinct, $\tau^f_{A,B,C,D,F}({\bf t})$ may also be expressed 
  in fermionic operator form as
 \be
 \tau^f_{(\tilde{A},B,\tilde{C},D,\tilde{F})}({\bf t}) = \langle 0 | \hat{\gamma}_+ ({\bf t}) \prod_{a=1}^lw_{a} (B,C)  
  \prod_{b=1}^l  v_{b}^{\dagger}(F, D)  | 0 \rangle,  \ee
 where  $\tilde{A}$,  $\tilde{C}$ are defined as in eq.~(\ref{Ktilde_Rtilde_def}), with $\tilde{A}$ satisfying the rank-1 condition
 (\ref{Atilde_rank1}) for ${\vf}$, ${\vg}$ defined in eq.~(\ref{fgKRdef}), $\tilde{F}$  defined by
 \be
 \tilde{F} = F (\tilde{K}^T(D))^{-1},
 \ee
 and the fermionic creation and annihilation operators $w_a$, $v_a^\dag$ are defined by
  \bea
 \label{w_b_C}
 w_a (B, C, D) &\&:= \sum_{j=1}^M\sum_{\nu=1}^{N_j} C_{a, (j,  \nu)} \Psi_D^{\nu} (\beta_j),  \\ 
v^\dag_{b} (F, D) &\& := \sum_{i=1}^m\sum_{\mu=1}^{n_i} F_{b,(i,  \mu)} (\Psi^{\mu} (\delta_i))^\dag,  \quad   a=1, \dots, l.
 \eea
 This may be shown using Wick's identity,  exactly as in the proof of Theorem \ref{theorem_fermionic_tau_BCD}.
Choosing $l=n$  and $\det F =1$ , the  case of eq.~(\ref{tau_BCD_gauge}) is recovered.

%%%%%%%%%%%%%%% Section 4  %%%%%%%%%%%%%%%%%%%

\section{Affine coordinates and Schur function expansions}

The aim of this section is to find the affine coordinates of the subspace $W\in Gr_{\HH_+}(\HH)$ corresponding to a given generalized Gekhtman-Kasman $\tau$-function $\tau^f_{(A,B,C,D,F)}(\tb)$ as defined in (\ref{tau_gen_GK}). This allows us to obtain the Pl\"ucker coordinates of $W$ through the Giambelli formula (\ref{gen_giambelli})
and hence the Schur function expansion of $\tau^f_{(A,B,C,D,F)}(\tb)$. In particular, this gives the Schur function expansion for $\tau$-functions of the original Gekhtman--Kasman form $\tau^f_{(A,B,C)}(\tb)$, up to an explicit gauge factor.

%%%%%%%%%%%%%%% Subsection 4.1  %%%%%%%%%%%%%%%%%%%

\subsection{Subspaces with geometric affine coordinates}

Let $n$ and $N$ be positive integers (with no assumption on their relative size) and consider a quintuple 
$(\vf, \vg, B, D, M)$ consisting of a pair of vectors
\be
({\vf}, {\vg}) \in \Cn^n \times \Cn^N,
\ee
a pair of square matrices
\be 
B\in \Mat{N }{N}, \quad D \in \Mat{n }{n}
\ee
whose eigenvalues are inside the unit disk in the complex plane, and a rectangular matrix
\be
M \in \Mat{N }{n}.
\ee

Consider the subspace $W(\vf, \vg, B, D, M)$ in the big cell of the Segal-Wilson Grassmannian defined by choosing
 the affine coordinates to have the special form
\be
\AA_{ij} = \vg^{T}B^jM (D^{T})^i\, \vf \qquad i,j=0,1,\dots.
\ee
This will be referred to as \emph{geometric affine coordinates}, since the dependence of $A_{ij}$ on $D$ and $B$ 
are given through the matrix-valued ``geometric sequences'' $\{(D^{T})^i\}_{i\in \Nn}$ and $\{B^j\}_{j\in \Nn}$. 
That is,
\be
W(\vf, \vg, B, D, M)=\mbox{span}\left\{z^i+\sum_{j=0}^{\infty}\left(\vg^{T}B^j M(D^{T})^i \vf\right) z^{-j-1}\right\}_{ i\in \Nn}.
\ee
Since the spectrum of $B$ is contained in the unit disk, the basis elements can be rewritten in terms of the resolvent of $B$ as
\be
z^i + \vg^{T}(z\,\Id{N}-B)^{-1}M(D^{T})^i \vf, \quad i=0,1,\dots
\ee
Therefore $W(\vf, \vg, B, D, M)$ is the graph of the operator
\bea
T\ \colon\ \HH_{+}&\&\to \HH_{-}\cr
\phi&\& \mapsto  \frac{1}{2\pi i}\oint_{|\zeta|=1}\left[\vg^{T}(z\,\Id{N}-B)^{-1}M(\zeta\, \Id{n}-D^{T})^{-1} \vf\right] \phi(\zeta)d\zeta.
\eea

\begin{theorem} 
We have the inclusions
\label{thm:gr_rat}
\be
\label{eq:poly_chain}
\mu_D(z)\HH_{+} \subset W(\vf, \vg, B, D, M) \subset \frac{1}{\mu_B(z)}\HH_{+},
\ee
where $\mu_D$ and $\mu_B$ are the minimal polynomials of the matrices $D$ and $B$, respectively.
\end{theorem}
\begin{proof} For any $\phi(z) \in \HH_{+}$ the function $[T\phi](z)$ is rational with  possible  poles only at the eigenvalues of $B$. The common denominator of the entries of the resolvent matrix $(z\,\Id{N}-B)^{-1}$ is equal to the minimal polynomial $\mu_B(z)$ and therefore

\be
[T\phi](z) \in \frac{1}{\mu_B(z)}\HH_{+}\ \quad \mbox{ for all } \phi(z) \in \HH_{+}.
\ee
On the other hand, if $p(z)$ is a polynomial divisible by the minimal polynomial $\mu_D(z)$, we have
\be
\mu_D(z)|p(z) \quad \Rightarrow \quad [Tp](z) = \vg^{T}(z\,\Id{N}-B)^{-1}M p(D^{T}) \vf=0.
\ee
By continuity,
\be
T(\mu_D(z)\phi(z)) = 0 \quad \mbox{ for all } \phi(z) \in \HH_{+},
\ee
since the polynomials form a dense subset in $\HH_{+}$. This clearly implies that
\be
\mu_D(z)\HH_{+} \subset W(\vf, \vg, B, D, M).
\ee
\end{proof}
The chain of inclusions  (\ref{eq:poly_chain}) evidently also implies the weaker one
\be
r_D(z)\HH_{+} \subset W(\vf, \vg, B, D, M) \subset \frac{1}{r_B(z)}\HH_{+},
\ee
where $r_D(z)$ and $r_B(z)$ stand for the characteristic polynomials of the matrices $D$ and $B$, respectively.
\begin{theorem} The $\tau$-function of $W(\vf, \vg, B, D, M)$ is given by
\bea
\label{eq:tau_geometric}
\nonumber
 &\&\tau_{W(\vf, \vg, B, D, M)(\tb)}\\
\nonumber
 &\&= \det\left(\Id{n} + e^{-\sum{t_i (D^{T})^i}}
\left(\frac{1}{2\pi i}\oint_{|z|=1} (z\,\Id{n}-D^{T})^{-1}\vf\vg^{T}(z\,\Id{N}-B)^{-1}e^{\sum t_i z^i}dz\right)M\right)_{n \times n}\\
\\
\nonumber
&\&= \det\left(\Id{N} + Me^{-\sum{t_i(D^{T})^i}}
\left(\frac{1}{2\pi i}\oint_{|z|=1} (z\,\Id{n}-D^{T})^{-1}\vf\vg^{T}(z\,\Id{N}-B)^{-1}e^{\sum t_i z^i}dz\right)\right)_{N\times N}.\\
\eea
\end{theorem}
\begin{proof} Consider the block decomposition
\be
\gamma_{+}({\tb}) = 
\pmatrix{
a({\tb}) &b({\tb})\cr
0 &d({\tb})
}
\ee
with respect to the splitting $\HH= \HH_{+}\oplus \HH_{-}$.
Following \cite{SW}, the $\tau$-function of $W(\vf, \vg, B, D, M)$ is given by
\be
\tau_{W(\vf, \vg, B, D, M)}(\tb)=\det\left(\Id{\mathcal H_+}+T({\tb})\right),
\ee
where
\be
T({\tb})\ \colon\ \HH_{+}\to \HH_{+}, \quad T({\tb}) := b({\tb})T a({\tb})^{-1}.
\ee
Assuming  $\phi \in \HH_{+}$,
\bea
\nonumber
\left[T({\tb})\phi\right](z) &\&= \vg^{T}\left[\frac{1}{2\pi i}\oint_{|\zeta|=1}e^{\sum t_i \zeta^i}(\zeta\,\Id{N}-B)^{-1}\frac{d\zeta}{\zeta-z} \right]M e^{-\sum t_i (D^{T})^{i}}\phi(D^{T})\vf\\
&\&= \vg^{T}\left(e^{\sum t_i z^i}\Id{N}-e^{\sum t_i B^{i}}\right)(z\,\Id{N}-B)^{-1}M e^{-\sum t_i (D^{T})^{i}}\phi(D^{T})\vf,
\eea
given that all the eigenvalues of $B$ are inside the unit circle. Therefore  the operator $T({\tb})$ may be written in factorized form as
\be
T({\tb}) = L_2({\tb})L_1({\tb})
\ee
where
\be
\begin{array}{rcl}
L_1({\tb})\ \colon \HH_{+} &\to& \Cn^{n}\\
\phi & \mapsto& e^{-\sum t_i (D^{T})^i}\phi(D^{T})\vf,
\end{array}
\ee
and
\be
\begin{array}{rcl}
L_2({\tb})\ \colon\ \Cn^{n} & \to& \HH_{+}\\
\vv & \mapsto&  \vg^{T}\left(e^{\sum t_i z^i}\Id{N}-e^{\sum_{i=1}
^{\infty} t_i B^{i}}\right)(z\,\Id{N}-B)^{-1}M \vv .
\end{array}
\ee
By inverting the order of the operators $L_1({\tb})$ and $L_2({\tb})$ we obtain the map
\bea
\nonumber
L_1({\tb})L_2({\tb})\ \colon \ \Cn^{n} &\to& \Cn^n\\
\nonumber
\vv &\mapsto& e^{-\sum t_i (D^{T})^i}\left[\frac{1}{2\pi i}\oint_{|z|=1}\!\! (z\,\Id{n}-D{^T})^{-1}\vf
\vg^{T}(z\,\Id{N}-B)^{-1}e^{\sum t_i z^i}dz\right] M\vv, \\
\eea
where the second term vanishes since the fact that  the eigenvalues of $D$ and $B$ are inside the unit circle implies that
\be
\frac{1}{2\pi i}\oint_{|z|=1} (z\,\Id{n}-D^{T})^{-1}\vf\vg^{T}(z\,\Id{l_1}-B)^{-1}dz = 0 .
\ee

Applying the Weinstein--Aronszajn identity, 
\bea
\nonumber
\tau_{W(\vf, \vg, B, D, M)}(t) &\&=  \det\left(\Id{\HH_{+}} + L_2({\tb})L_1({\tb})\right)\\
\nonumber
&\&=  \det\left(\Id{n} + L_1({\tb})L_2({\tb})\right)\\
\nonumber
&\&=  \det\left(\Id{n} + e^{-\sum t_i (D^{T})^i}\left(\frac{1}{2\pi i}\oint_{|z|=1}\!\! (z\,\Id{n}-D{^T})^{-1}\vf
\vg^{T}(z\,\Id{N}-B)^{-1}e^{\sum t_i z^i}dz\right) M\right)\\
\nonumber
&\&=  \det\left(\Id{N} + Me^{-\sum t_i (D^{T})^i}\left(\frac{1}{2\pi i}\oint_{|z|=1}\!\! (z\,\Id{n}-D{^T})^{-1}\vf
\vg^{T}(z\,\Id{N}-B)^{-1}e^{\sum t_i z^i}dz\right) \right)\\
\eea
\end{proof}

The generalized Giambelli formula (\ref{gen_giambelli}) implies that the Pl\"ucker coordinates associated to $W(\vf, \vg, B, D, M)$ are
\be
\pi_{(a_1,\dots, a_k|b_1,\dots, b_k)}(W(\vf, \vg, B, D, M)) = (-1)^{\sum_{i=1}^{k}b_i}\det\left(\vg^{T}B^{b_j}M(D{^T})^{a_i}\vf\right)_{i,j = 1}^k,
\ee
Therefore the $\tau$-function $\tau_{W(\vf, \vg, B, D, M)}(\tb)$ has following Schur function expansion.

\begin{corollary}
\label{schur_expansion1}
\be
\tau_{W(\vf, \vg, B, D, M)}(\tb) = \sum_{(\ab|\bb)}\det\left(\vg^{T}(-B)^{b_j}M(D^{T})^{a_i}\vf\right)_{i,j=1}^k S_{(\ab|\bb)}(\tb).
\ee
\end{corollary}

%%%%%%%%%%%%%%% Subsection 4.2  %%%%%%%%%%%%%%%%%%%

\subsection{Specialization to $\tau^f_{(A,B,C,D,F)}({\bf t})$}
We now specialize the  above  to the subspaces corresponding to 
 $\tau$-functions of the generalized Gekhtman-Kasman form $\tau^f_{(A,B,C,D,F)}({\bf t})$ 
as defined in  eq.~(\ref{tau_gen_GK}).

We first note that the rank-1 condition
\be
AB - D^{T}A  = \vf \vg^{T}
\label{eq:rank1_again}
\ee
is equivalent to the resolvent identity
\be
\label{eq:resolvents}
A(z\,\Id{N}-B)^{-1} = (z\,\Id{n}-D^{T})^{-1}A +(z\,\Id{n}-D^{T})^{-1}\vf\vg^{T}(z\,\Id{N}-B)^{-1}.
\ee
More generally,  equation (\ref{eq:rank1_again}) implies that for any function $\phi(z)$ holomorphic in the unit disk we have
\be
A\phi(B) = \phi(D^{T})A +\frac{1}{2\pi i}\oint_{|\zeta|=1}(\zeta\, \Id{n}-D^{T})^{-1}\vf\vg^{T}(\zeta\, \Id{N}-B)^{-1}\phi(\zeta)d\zeta.
\ee
In particular,
\be
\label{decoupled_int}
\frac{1}{2\pi i}\oint_{|\zeta|=1} (\zeta\, \Id{n}-D^{T})^{-1}\vf\vg^{T}(\zeta\, \Id{N}-B)^{-1}e^{\sum t_i \zeta^i}d\zeta= Ae^{\sum t_i B^i} - e^{\sum t_i( D^{T})^i}A,
\ee
for any $e^{\sum t_i z^{i}} \in \Gamma_{+}$.
\begin{theorem}
Let $A \in \Mat{n }{N}, B \in \Mat{N }{N}, D \in \Mat{n }{n}, \vf \in \Cn^{n}, \vg \in \Cn^{N}$ be
such that the rank-$1$ condition (\ref{eq:rank1_again})
holds and $F \in\Mat{l }{n}, C\in \Mat{l }{N}$ satisfy the non-singularity condition
\be
\det(FAC^{T})\not=0.
\ee
The  $\tau$-function $\tau^{f}_{(A,B,C,D,F)}({\tb}) $ is 
\be
\tau^{f}_{(A,B,C,D,F)}({\tb}) = \det(FAC^{T}) \tau_{W(\vf, \vg, B, D, M)}({\tb}) ,
\ee
where $W(\vf, \vg, B, D, M)$ is the subspace with geometric affine coordinates, associated to the quintuple $(\vf, \vg,B,D,M)$, with $M$ given by
\be
M = C^{T}(FAC^{T})^{-1}F.
\ee
\end{theorem}
\begin{proof}
As a consequence of the decoupling integral formula (\ref{decoupled_int}), the first finite determinant expressing $\tau_{W(\vf, \vg, B, D, M)}(\tb)$ in  (\ref{eq:tau_geometric}) simplifies to the form
\bea
\nonumber
\tau_{W(\vf, \vg, B, D, C^{T}(FAC^{T})^{-1}F)}(\bm t) \& &= \det\left(\Id{n} + e^{-\sum{t_i(D^{T})^i}}\left(Ae^{\sum t_i B^i} - e^{\sum t_i (D^{T})^i}A\right) C^{T}(FAC^{T})^{-1}F\right)_{n\times n}\\
\nonumber
\& &= \det\left(\Id{l} + Fe^{-\sum{t_i(D^{T})^i}}\left(Ae^{\sum t_i B^i} - e^{\sum t_i (D^{T})^i}A\right) C^{T}(FAC^{T})^{-1}\right)_{l\times l}\\
\& &= \det(FAC^{T})^{-1}\det\left(Fe^{-\sum{t_i(D^{T})^i}}Ae^{\sum t_i B^i}C^{T}\right),
\eea
where the Weinstein--Aronszajn identity was used in the first equality.
\end{proof}
\begin{remark} As an alternative representation of the  $\tau$-function $\tau^{f}_{(A,B,C,D,F)}({\tb})$, we can choose the subspace $W(\vg,\vf,D^{^T},B^{T}, - F^{T}(CA^TF^{T})^{-1}C)$,
 for which
\be
\tau^{f}_{(A,B,C,D,F)}({\tb}) = \det(FAC^{T}) \tau_{W(\vg,\vf,D^{^T},B^{T}, - F^{T}(CA^TF^{T})^{-1}C)}(-{\tb}) .
\ee
\end{remark}
From Corollary \ref{schur_expansion1} follows
\begin{corollary} The following Schur function expansion holds:
\be
\label{GK_Schur}
\tau^{f}_{(A,B,C,D,F)}(\tb) = \sum_{(\ab|\bb)} \det(\vg^{T}(-B)^{b_j} C^{T}(FAC^{T})^{-1}F(D^{T})^{a_i}\vf)S_{(\ab|\bb)}(\tb).
\ee
\end{corollary}
By Theorem~\ref{thm:gr_rat} we have the inclusions
\begin{corollary} 
\bea
r_D(z) \HH_+  \subset \mu_D(z)\HH_{+} \subset W(\vf, \vg, B, D, M) \subset \frac{1}{\mu_B(z)}\HH_{+}\subset \frac{1}{r_B(z)}\HH_{+}, \nonumber\\
\\\label{eq:GK_inclusions}
r_B(z)\HH_{+} \subset \mu_B(z)\HH_{+} \subset W(\vg,\vf,D^{^T},B^{T}, - E^{T}(CAF^{T})^{-1}C)\subset \frac{1}{\mu_D(z)}\HH_{+} \subset \frac{1}{r_D(z)}\HH_{+}.\nonumber\\
\eea
\end{corollary}
Specializing to the case when $l=n$ and $F=\Ib_n$, we obtain:
\begin{corollary}The Gekhtman--Kasman $\tau$-function 
\be
\tau^f_{(A,B,C)}(\tb)=\det\left(Ae^{\sum t_i B^i}C^{T}\right)
\ee
with nonvanishing constant term 
\be 
\det(AC^{T}) \neq 0 
\ee
 can be written in terms of the $\tau$-functions of either of the subspaces $W(\vf, \vg, B, D, C^{T}(AC^{T})^{-1})$ or $W(\vg,\vf,D^{^T},B^{T}, - (CA^{T})^{-1}C)$ as
\be
\label{GK_W}
\tau^{f}_{(A,B,C)}(\tb)=
\left\{
\begin{array}{l}
\ds e^{\sum {t_i{\Tr}(D^i)}}\det(AC^{T})\tau_{W(\vf, \vg, B, D, C^{T}(AC^{T})^{-1})}(\tb),\\
\\
\ds e^{\sum {t_i{\Tr}(D^i)}}\det(AC^{T})\tau_{W(\vg,\vf,D^{^T},B^{T}, - (CA^TF^{T})^{-1}C)}(-\tb) .
\end{array}
\right.
\ee
\end{corollary}

The presence of the gauge factor $e^{\sum {t_i{\Tr}(D)^i}}$ in (\ref{GK_W}), when compared with the linear exponential factor appearing in (\ref{gaugetransform}), implies that the effective subspace $\tilde W(\vf, \vg, B, D, C^{T}(AC^{T})^{-1})$ that corresponds to $\tau^{f}_{(A,B,C)}(\tb)$ is given by
\be
\tilde W(\vf, \vg, B, D, C^{T}(AC^{T})^{-1}):=\frac{z^n}{r_D(z)}W(\vf, \vg, B, D, C^{T}(AC^{T})^{-1}),
\ee
which satisfies the inclusion chain
\be
z^n\HH_{+} \subset \tilde W(\vf, \vg, B, D, C^{T}(AC^{T})^{-1}) \subset \frac{z^{n}}{r_B(z)r_D(z)}\HH_{+}.
\ee
The subspace $\tilde W(\vf, \vg, B, D, C^{T}(AC^{T})^{-1})$ thus belongs to the sub-Grassmannian $\Gr_1$ 
of ref. (\cite{SW}).

%%%%%%%%%%%%%%% Section 5  %%%%%%%%%%%%%%%%%%%
\section{Concluding remarks} We have shown how KP $\tau$-functions of the 
finite determinantal  form (\ref{tau_GK}) and (\ref{tau_gen_GK}) may be derived
through a subquotienting procedure applied to KP flows on the infinite dimensional 
Grassmannians of  Sato, Segal and Wilson \cite{S, SS, SW},  viewed as projections to
linear exponential flows on finite Grassmannians. These solutions were also expressed
as  fermionic operator  vacuum expectation values of suitably defined products
of creation and annihilation operators. This approach may be applied
more generally to other integrable hierarchies, such as discrete KP, 
 MKP chains of $\tau$-functions, depending both on continuous and discrete flow variables,
and the 2-Toda lattice hierarchy, resulting in analogous finite determinantal  $\tau$-functions.

The cases treated here are all ``generic'' solutions to the rank-1 condition (\ref{rank1_D}),
in the sense that the eigenvalues of the pair of matrices $B$ and $D$ are required to be 
distinct, and the generic conditions (\ref{eq:generic_vec}) satisfied. In the sequel \cite{BFH}, these
conditions will dropped, and a complete classification of all solutions of the rank-1 condition (\ref{rank1_D}),
for all possible pairs of vectors ${\bf f}, {\bf g}$ and matrices $B$, $D$ determined, up to natural 
equivalence within orbits of the stabilizers of $B$ and $D$.

\bigskip

%%%%%%%%%%%%%%%%%%%%%%% Appendix %%%%%%%%%%%%%%%%%%
 \appendix
 \section{Appendix: Proof that $\tau^f_{(A,B,C,D,F)}({\bf t})$ satisfies\\ the Hirota bilinear equations}
  \label{AppendA}
  
     A concise way to express the Hirota bilinear equations for a KP $\tau$-function $\tau({\bf t})$ is to define
 an associated family of 2-forms $\xi(z_1, z_2, z_3, z_4, {\bf t})$ in $\Cn^4$, depending on four complex parameters $(z_1, z_2, z_3, z_4)$ together with the  KP flow parameters ${\bf t} = (t_1, t_2, \dots )$ : 
 \be
  \xi(z_1, z_2, z_3, z_4, {\bf t})  = \sum_{i,j = 1}^4 \xi_{ij} e_1 \wedge e_j \in \Lambda^2 \Cn^4,
 \ee
where
 \be 
 \xi_{ij} := (z_i -z_j) \tau ({\bf t} - [z_i^{-1}] - [z_j^{-1}]).
 \label{xi_def}
 \ee
    The Hirota bilinear equations are then equivalent \cite{S, SS} to the  single Pl\"ucker relation defining the image of the Grassmannian $\Gr_{2}(\Cn^4)$ in $\Pb (\Lambda^2 \Cn^4) $ under the Pl\"ucker map:
\bea
\grP\grl:\Gr_{2}(\Cn^4)&\&\ra \Pb (\Lambda^2 \Cn^4) \cr
\grP\grl: \span(W_1, W_2) &\& \mapsto [W_1 \wedge W_2], \quad W_1, \ W_2 \in \Cn^4.
\eea
i.e. the decomposibility condition
\bea
\xi \wedge \xi  &\&= 0, \\
\xi_{12} \xi_{34} - \xi_{13}\xi_{24} + \xi_{14}\xi_{23}  &\& = 0,
\label{24_plucker}
\eea
 satisfied identically in the parameters $(z_1, z_2, z_3, z_4)$ for all ${\bf t}$. We will prove this 
holds for the $\tau$-function $\tau^f_{(A,B,C,D,F)}({\bf t})$ by explicitly computing the vectors $W_1, W_2 \in \Cb^4$.

Define $\xi^{(A,B,C,D,F)}(z_1, z_2, z_3, z_4, {\bf t})$ as in (\ref{xi_def}) with $\tau = \tau^f_{(A,B,C,D,F)}$
\be
\xi^{(A,B,C,D,F)}_{ij}  (z_1, z_2, z_3, z_4, {\bf t}) :=  (z_i -z_j) \tau^f_{(A,B,C,D,F)} ({\bf t} - [z_i^{-1}] - [z_j^{-1}]),
\ee
and also define
\bea
F({\bf t}) &\& :=  F e^{-\sum_{i=1}^\infty t_i(D^T)^i}, \quad  C^T({\bf t}) := e^{\sum_{i=1}^\infty t_iB^i} C^T,  \cr
M^T({\bf t}) &\& := C^T({\bf t}) \left( F ({\bf t} )A C^T({\bf t})\right)^{-1}F({\tb}),
\eea
which is well defined provided $\tau^f_{(A,B,C,D,F)} ({\bf t})\ne 0$.
The following lemma shows that eq.~(\ref{24_plucker}) is satisfied by $\tau^f_{(A,B,C,D,F)}({\bf t})$. 
\bl 
\be
\xi^{(A,B,C,D,F)} = W_1\wedge W _2
\label{xiABCDFW1W2}
\ee
where
\bea
 W_1 &\& := \tau^f_{(A,B,C,D,F)} (H(z_1), H(z_2), H(z_3), H(z_4))  \cr
 W_2 &\& := (G(z_1), G(z_2), G(z_3), G(z_4))
 \label{ABCDE_W_24_def}
\eea
with the rational functions $H(z)$, $G(z)$ defined by
\bea
H(z) &\&:=  \tau^f_{(A,B,C,D,F)}({\bf t}) \left( 1 +  {\vg}^t M({\bf t}) (D^T- z\, \Id{n})^{-1} {\vf}\right) \cr
G(z) &\&:= z  +  {\vg}^t B M({\bf t}) ( D^T - z\, \Id{n})^{-1} {\vf} .
\label{HG_def}
\eea
\el
\begin{proof}
Since
\be 
e^{-\sum_{i=1}^\infty \frac{1}{i} (B/z)^i}= \Id{N} - B/z , \quad e^{\sum_{i=1}^\infty \frac{1}{i}(D^T/z)^i} = (\Id{n} -D^T/z)^{-1} ,
\ee
we have
\bea
\xi^{(A,B,C,D,F)}_{ij}&\&= (z_i - z_j)\det\left(F({\bf t}) (D^T-z_i\, \Id{n})^{-1} (D^T-z_j\, \Id{n})^{-1} A  (B-z_i\, \Id{N}) (B-z_j\, \Id{N}) C^T({\bf t})\right) \cr
&\& = (z_i - z_j) \det( F({\bf t}) A C^T({\bf t})\cr
&\& \quad + (D^T-z_i\, \Id{n})^{-1} (D^T-z_j\, \Id{n})^{-1} (D^T - z_i - z_j) {\vf} {\vg}^T C^T({\bf t}) 
\cr
&\& \quad +  (D^T-z_i\, \Id{n})^{-1} (D^T-z_j\, \Id{n})^{-1}  {\vf} {\vg}^T B C^T({\bf t}) ) \cr
&\& = (z_i - z_j) \det( F({\bf t}) A C^T({\bf t})) \det ( \Id{l}\cr &\& \quad +  (D^T-z_i\, \Id{n})^{-1} (D^T-z_j\, \Id{n})^{-1} (D^T - z_i - z_j) {\vf} 
{\vg}^T M^T({\bf t}) \cr
&\& \quad +  (D^T-z_i\, \Id{n})^{-1} (D^T-z_j\, \Id{n})^{-1}  {\vf} {\vg}^T B M^T({\bf t} ))\cr
&\& = (z_1 -z_2) \tau^f_{(A,B,C,D,F)}({\bf t})\det\left(  \Id{l }+ {\vf}_1^T(z_1, z_2) {\vg}^T_1 
+ {\vf}_2^T(z_1, z_2) {\vg}^T_2 \right)
\label{xiij_calc}
\eea
where
\bea
{\vf}_1 &\&:= (D^T-z_i\, \Id{n})^{-1} (D^T-z_j\, \Id{n})^{-1} (D^T - z_i - z_j) {\vf} , \quad
{\vg}^T_1:= {\vg}^T M^T({\bf t}) , \cr
{\vf}_2 &\& := (D^T-z_i\, \Id{n})^{-1} (D^T-z_j\, \Id{n})^{-1} {\vf}, \qquad \qquad  \qquad \quad \  
{\vg}^T_2:= {\vg}^T B M^T({\bf t}).
\eea
In the second line of (\ref{xiij_calc}) we have used the rank-$1$ condition (\ref{rank1_D}) in the form
\be
A (B - z\,\Id{N}) = (D- z \,\Id{n}) A + {\vf} {\vg}^T.
\ee
 Using the Weinstein-Aronszajn identity, we can rewrite this as a $2 \times 2 $ determinant
 \bea
 \xi^{(A,B,C,D,F)}_{ij} &\&= (z_1 -z_2) \tau^f_{(A,B,C,D,F)}({\bf t}) 
 \det \pmatrix {1 + {\vg}^T_1 {\vf}_1 & {\vg}^T_1 {\vf}_2 \cr
  {\vg}^T_2 {\vf}_1  & 1+  {\vg}^T_2 {\vf}_2 }  \cr
  &\& \cr
&\& = \tau^f_{(A,B,C,D,F)}({\bf t})  \det\pmatrix { H(z_i) & H(z_j) \cr G(z_i) & G(z_j)},
 \eea
 where the last line follows from elementary column operations.
This is equivalent to Eq.~(\ref{xiABCDFW1W2}).
\end{proof}

%%%%%%%%%%%%%%%  Bibliography  %%%%%%%%%%%%%%%%%%%

\noindent
{\bf Acknowledgments.}  This work was begun while T.~D.~F. and F.~B. were postdoctoral fellows at the Centre de recherches math\'ematiques (CRM), Montr\'eal, and completed while F.~B. was at SISSA,  Trieste and T.~D.~F. was at  LAPTh Laboratoire d'Annecy-le-Vieux, France. Work of J.~H. was supported in part by the Natural Science and Engineering Research Council at Canada (NSERC) and the Fonds Qu\'eb\'ecois de la recherche sur la nature et les technologies (FQRNT). 

\appendix

\end{document}